\documentclass{ecai}

\usepackage{latexsym}
\usepackage{amssymb}
\usepackage{amsmath}
\usepackage{amsthm}
\usepackage{booktabs}
\usepackage{graphicx}
\usepackage{color}

\usepackage{tikz}
\usepackage{graphicx}
\usetikzlibrary{shapes,arrows,automata}

\usepackage{mathpartir}

\usepackage{amsmath}
\usepackage{amsfonts}
\usepackage{amsthm}
\usepackage{mathtools}
\usepackage{thmtools}

\usepackage{mathrsfs}  
\usepackage{stmaryrd}
\usepackage{thm-restate}
\usepackage{environ}

\usepackage{listings}
\usepackage[inline]{enumitem}

\usepackage{prftree}
\usepackage{upgreek}
\usepackage{wasysym}
\usepackage{subcaption}
\usepackage{hyperref}
\usepackage{cancel}

\usepackage{xcolor}
\usepackage{booktabs}

\usepackage{eurosym} %

\makeatletter
\newcommand*{\da@rightarrow}{\mathchar"0\hexnumber@\symAMSa 4B }
\newcommand*{\da@leftarrow}{\mathchar"0\hexnumber@\symAMSa 4C }
\newcommand*{\xdashrightarrow}[2][]{%
  \mathrel{%
    \mathpalette{\da@xarrow{#1}{#2}{}\da@rightarrow{\,}{}}{}%
  }%
}
\newcommand{\xdashleftarrow}[2][]{%
  \mathrel{%
    \mathpalette{\da@xarrow{#1}{#2}\da@leftarrow{}{}{\,}}{}%
  }%
}
\newcommand*{\da@xarrow}[7]{%
  \sbox0{$\ifx#7\scriptstyle\scriptscriptstyle\else\scriptstyle\fi#5#1#6\m@th$}%
  \sbox2{$\ifx#7\scriptstyle\scriptscriptstyle\else\scriptstyle\fi#5#2#6\m@th$}%
  \sbox4{$#7\dabar@\m@th$}%
  \dimen@=\wd0 %
  \ifdim\wd2 >\dimen@
    \dimen@=\wd2 %
  \fi
  \count@=2 %
  \def\da@bars{\dabar@\dabar@}%
  \@whiledim\count@\wd4<\dimen@\do{%
    \advance\count@\@ne
    \expandafter\def\expandafter\da@bars\expandafter{%
      \da@bars
      \dabar@ 
    }%
  }%
  \mathrel{#3}%
  \mathrel{%
    \mathop{\da@bars}\limits
    \ifx\\#1\\%
    \else
      _{\copy0}%
    \fi
    \ifx\\#2\\%
    \else
      ^{\copy2}%
    \fi
  }%
  \mathrel{#4}%
}
\makeatother

\newtheorem{theorem}{Theorem}
\newtheorem{lemma}[theorem]{Lemma}

\newtheorem{proposition}[theorem]{Proposition}

\newtheorem{definition}{Definition}
\newtheorem{example}{Example}

\newcommand{\BibTeX}{B\kern-.05em{\sc i\kern-.025em b}\kern-.08em\TeX}

\renewcommand{\theparagraph}{\thesubsection\alph{paragraph}}         %
\renewcommand{\tablename}{Table}
\renewcommand{\figurename}{Figure}

\newcommand{\viola}[1]{}

\lstdefinelanguage{MuAC}{
		keywords = {Gives, Me, Requester, with},
	    comment=[l]{//},
	}

\lstdefinestyle{MuAC}{
		language=MuAC,
		basicstyle=\footnotesize\ttfamily, 
		otherkeywords={:-},
		keywordstyle=\color{blue},
		commentstyle=\color{gray},
	}

\def\makenewenum#1#2{%
\newcounter{cnt#1}
\newenvironment{#1}%
{\begin{list}{\makebox[0pt][r]{#2}}%
{\setlength{\itemsep}{0pt}%
 \setlength{\parsep}{.2em}%
 \setlength{\leftmargin}{2em}%
 \setlength{\labelwidth}{.2em}%
 \usecounter{cnt#1}}}%
{\end{list}}}
\makenewenum{enu}{\arabic{cntenu}.}
\makenewenum{enum}{(\arabic{cntenum})}
\makenewenum{itenum}{\em(\arabic{cntenum})}
\makenewenum{en}{(\roman{cnten})}
\makenewenum{renum}{\textit{(\roman{cntrenum})}}
\makenewenum{enhyphen}{(\roman{cntenhyphen}')}
\makenewenum{aenum}{(\alph{cntaenum})}
\makenewenum{itemizes}{$\bullet$}

\newenvironment{restate-theorem}[1]%
  {\begin{trivlist}\item[]{\normalsize\bfseries{\sffamily}
        Restatement of Theorem~#1.}\hspace*{0mm}\it}%
  {\end{trivlist}}

\newenvironment{restate-lemma}[1]%
  {\begin{trivlist}\item[]{\normalsize\bfseries{\sffamily}
        Restatement of Lemma~#1.}\hspace*{0mm}\it}%
  {\end{trivlist}}

\newenvironment{restate-proposition}[1]%
  {\begin{trivlist}\item[]{\normalsize\bfseries{\sffamily}
        Restatement of Proposition~#1.}\hspace*{0mm}\it}%
  {\end{trivlist}}

\newcommand{\tr}{tr}

\newcommand{\semden}[1]{\llparenthesis\, #1 \,\rrparenthesis}

\def\makenewenum#1#2{%
	\newcounter{cnt#1}
	\newenvironment{#1}%
	{\begin{list}{\makebox[0pt][r]{#2}}%
			{\setlength{\itemsep}{0pt}%
				\setlength{\parsep}{.2em}%
				\setlength{\leftmargin}{2em}%
				\setlength{\labelwidth}{.2em}%
				\usecounter{cnt#1}}}%
		{\end{list}}}

\NewEnviron{smalign*}{%
			\begin{align*}
			\BODY
			\end{align*}}

\newcommand{\semantics}[1]{\llparenthesis#1\rrparenthesis}

\DeclareMathOperator{\linearcontract}{\multimap\hspace{-8pt}\multimap}

\newcommand{\contract}{\twoheadrightarrow}

\newcommand{\MuACL}{CEL}

\newcommand{\MuACstate}{\mathit{st}}
\newcommand{\st}{\MuACstate}
\newcommand{\MuACStates}{\mathit{St}}
\newcommand{\usr}{\mathit{a}}
\newcommand{\Usr}{\mathcal{A}}
\newcommand{\coal}{\mathit{C}}
\newcommand{\res}{\mathit{r}}
\newcommand{\Res}{\mathcal{R}}
\newcommand{\Tran}{\mathit{Tr}}
\newcommand{\tran}{\mathit{tr}}
\newcommand{\exc}{\mathit{exc}}
\newcommand{\Exc}{\mathit{Exc}}

\newcommand{\MuACLs}{\MuACL$_{\text{(Cut)}}$}

\newcommand{\Pol}{\mathit{Pol}}
\newcommand{\pol}{\mathit{pol}}

\newcommand{\Alice}{\texttt{A}}
\newcommand{\Bob}{\texttt{B}}
\newcommand{\Charlie}{\texttt{C}}
\newcommand{\ki}{\texttt{k}}
\newcommand{\lm}{\texttt{l}}
\newcommand{\ma}{\texttt{m}}

\begin{document}

\begin{frontmatter}

\paperid{2843}

\title{A Logic for Policy Based Resource Exchanges in Multiagent Systems}

\author[A]{\fnms{Lorenzo}~\snm{Ceragioli}
}
\author[A,B]{\fnms{Pierpaolo}~\snm{Degano}
}
\author[A]{\fnms{Letterio}~\snm{Galletta}
}
\author[C]{\fnms{Luca}~\snm{Viganò}
} 

\address[A]{IMT School for Advanced Studies Lucca, Italy}
\address[B]{Dipartimento di Informatica, University of Pisa, Italy} %
\address[C]{Department of Informatics, King's College London, UK}

\begin{abstract}
In multiagent systems autonomous agents interact with each other to achieve individual and collective goals. 
Typical interactions concern negotiation and agreement on resource exchanges.
Modeling and formalizing these agreements pose significant challenges, particularly in capturing the dynamic behaviour of agents, while ensuring that resources are correctly handled.
Here, we propose exchange environments as a formal setting where agents specify and obey exchange policies, which are declarative statements about what resources they offer and what they require in return.
Furthermore, we introduce a decidable extension of the computational fragment of linear logic as a fundamental tool for representing  exchange environments 
and studying their dynamics in terms of provability.

\end{abstract}

\end{frontmatter}

\section{Introduction}\label{sec:intro}

Multiagent systems represent complex environments where autonomous agents interact to achieve individual and collective goals. 
Central to these interactions is the concept of resource exchange, which requires agents to negotiate and reach an agreement about the allocation of resources. 
Modeling and formalizing these agreements pose significant challenges to capture the behaviour of agents, in particular for ensuring that resources are allocated profitably for the agents involved in exchanges.
Besides AI, the study of such models has an effect on a broad spectrum of research domains, ranging from cooperative problem-solving to sharing economy scenarios. 
In cooperative problem-solving, the allocation of resources to agents is essential for enabling each agent to fulfill its assigned tasks effectively. 
In a typical sharing economy scenario, a community of users rely on a digital platform to foster collaboration and to share and transfer to each other resources and assets via peer-to-peer transactions.
Dynamic resource management has been addressed from diverse perspectives, including the design of negotiation and optimization strategies, game-theoretical analysis and logics.
In particular, the last approach emphasizes the development of frameworks to represent and study resource allocation and negotiation in multiagent systems, leveraging logic as a fundamental tool for this purpose.
Here, we follow the logical approach, and provide a twofold contribution.
First, we introduce the notion \emph{exchange environment} to formalize a multiagent system that aims to model both cooperative and competitive behavior. 
An exchange environment is a transition system, where states record the ownership of the resources, and transitions represent resource exchanges between agents, who can form coalitions.
Furthermore,  exchanges are constrained by declarative statements, called \emph{exchange policies}, which agents specify in isolation to prescribe what resources they offer and what they require in return (examples are in~\autoref{sec:examples}).
Using such policies, agents regulate competitiveness and foster cooperation.
The exchanges in a transition must guarantee that each participant gives the promised resources and gets the required ones, namely it is an \emph{agreement}.

In addition, we consider agent's \emph{evaluation functions}, and characterize when each permitted exchange is beneficial to all the agents of a coalition, namely it is a \emph{deal}.
Our second contribution is checking that a resource exchange is a deal.
To achieve that, it suffices to inspect all the rules of a policy and to check that the utility value of the obtained resources is greater than that of those given away.

Moreover, a crucial issue arises when verifying that the policies of all the participants are met, i.e. the exchange is an agreement.
Agreements may be circular, as it is typical of human and of virtual contracts.
To address this issue, we extend the computational fragment of linear logic to obtain \emph{Computational Exchange Logic}, \MuACL\ for short.
This logic handles circularity through a specific operator, called \emph{linear contractual implication}, inspired by PCL~\cite{BZ}.
To the best our knowledge, \MuACL\ is the first logic that combines linear and contractual aspects.
Every exchange $\exc$ is then encoded as a \MuACL\ formula, and verifying that $\exc$ is  indeed an agreement is reduced to proving the corresponding formula.
This procedure is effective because the validity of \MuACL\ formulas is decidable, which is  another main technical result of ours.

Another advantage of \MuACL\ is that it provides us with the means to consider as legal sequences of exchanges where an agents can contract temporary debts by offering resources that they do not currently have, but will acquire in a subsequent exchange.
Remarkably, the cut rule of the logic suffices to handle debts.

We proceed as follows.
\autoref{sec:examples} presents a running example.
We formalize exchange environments, policies and their relation with valuation functions in~\autoref{sec:env}.
Our logic \MuACL\ is in~\autoref{sec:col-formal}, and its extension to deal with debts in~\autoref{subsec:fair-computations}.
Finally, \autoref{sec:col-related} presents related works and \autoref{sec:col-conclude} draws conclusions.
The appendices contain the full proofs.

\section{Working examples and exchange policies}\label{sec:examples}
\begin{figure*}[t]
\begin{subfigure}{0.2\textwidth}
	\centering
	\footnotesize
	\begin{tikzpicture}
		\node[inner sep=0pt] (a)  at (1.75,.66) {\begin{tabular}{c}\includegraphics[scale=0.04]{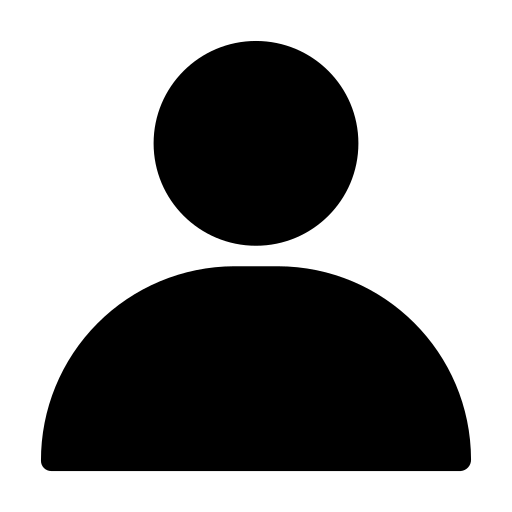}
				\\[-0.02cm] Alice
		\end{tabular}};
		
		\node[inner sep=0pt] (b)  at (3.5,0) {\begin{tabular}{c}\includegraphics[scale=0.04]{images/user}
				\\[-0.02cm] Bob
		\end{tabular}};

		{
			\draw[o->, thick, bend right] (a) edge node[above]{$l$} (b);
			\draw[o->, thick, bend right] (b) edge node[above]{$k$} (a);
		}
	\end{tikzpicture}
	\caption{Direct exchange.}
	\label{fig:ex:1}
\end{subfigure}
\begin{subfigure}{0.25\textwidth}
	\centering
	\footnotesize
	\begin{tikzpicture}
		\node[inner sep=0pt] (alice)  at (0,0) {\begin{tabular}{c}\includegraphics[scale=0.04]{images/user}
				\\[-0.02cm] Alice
		\end{tabular}};
		
		\node[inner sep=0pt] (bob)  at (1.5,.66) {\begin{tabular}{c}\includegraphics[scale=0.04]{images/user}
				\\[-0.02cm] Bob
		\end{tabular}};
		
		\node[inner sep=0pt] (carl)  at (3,0) {\begin{tabular}{c}\includegraphics[scale=0.04]{images/user}
				\\[-0.02cm] Carl
		\end{tabular}};	
		
		{
			\draw[o->, thick, bend left] (alice) edge node[above]{$k$} (bob);
			\draw[o->, thick, bend left] (bob) edge node[above]{$m$} (carl);
			\draw[<-o, thick, out=-20, in=-160] (alice) edge node[above]{$l$} (carl);
		}
	\end{tikzpicture}
	\caption{Circular exchange.}
	\label{fig:ex:2}
\end{subfigure}
\begin{subfigure}{0.25\textwidth}
	\centering
	\footnotesize
	\begin{tikzpicture}
		\node[inner sep=0pt] (bob)  at (0,0) {\begin{tabular}{c}\includegraphics[scale=0.04]{images/user}
				\\[-0.02cm] Bob
		\end{tabular}};
		
		\node[inner sep=0pt] (alice)  at (1.5,.66) {\begin{tabular}{c}\includegraphics[scale=0.04]{images/user}
				\\[-0.02cm] Alice
		\end{tabular}};
		
		\node[inner sep=0pt] (carl)  at (3,0) {\begin{tabular}{c}\includegraphics[scale=0.04]{images/user}
				\\[-0.02cm] Carl
		\end{tabular}};	
		
		{
			\draw[o->, thick, bend left] (bob) edge node[above]{$k$} (alice);
			\draw[<-o, thick, out=-20, in=-160] (bob) edge node[above]{$l$} (carl);
		}
	\end{tikzpicture}
	\caption{Coalition payment.}
	\label{fig:ex:3}
\end{subfigure}
\begin{subfigure}{0.25\textwidth}
	\centering
	\footnotesize
	\begin{tikzpicture}
		\node[inner sep=0pt] (carl)  at (0,0) {\begin{tabular}{c}\includegraphics[scale=0.04]{images/user}
				\\[-0.02cm] Carl
		\end{tabular}};
		
		\node[inner sep=0pt] (alice)  at (2,.66) {\begin{tabular}{c}\includegraphics[scale=0.04]{images/user}
				\\[-0.02cm] Alice
		\end{tabular}};
		
		\node[inner sep=0pt] (bob)  at (4,0) {\begin{tabular}{c}\includegraphics[scale=0.04]{images/user}
				\\[-0.02cm] Bob
		\end{tabular}};	
		
		{
			\draw[<-o, thick, bend left] (bob) edge node[above]{$l$} (alice);
			\draw[o->, thick, bend right,out=-30] (bob) edge node[above]{$2m$} (alice);
			\draw[o->, thick, bend left] (alice) edge node[above]{$m$} (carl);
			\draw[<-o, thick, out=-20, bend right] (alice) edge node[above]{$l$} (carl);
		}
	\end{tikzpicture}
	\caption{Debit.}
	\label{fig:ex:4}
\end{subfigure}
\vspace{\baselineskip}
\caption{Examples of agreements among players.}
\vspace{\baselineskip}
\label{fig:ex:agreements}
\end{figure*}
Consider three agents Alice, Bob and Carl (abbreviated \Alice, \Bob\ and \Charlie).
Let their set of resources consists of kiwis, lemons and mandarins (written \ki, \lm\ and \ma).
Assume agents freely form a coalition and interact by exchanging resources among them to obtain others more valuable for them (some examples of exchanges are in \autoref{fig:ex:agreements}.).
Instead of directly bargaining one with the others, agents define their policies in advance  
whose application is automatic: 
an agent accepts to perform an exchange if and only if it obeys her policy.
Our goal is twofold: $(1)$ to propose a model rich enough to handle the circular reasoning needed for resolving the conditions expressed by the policies; and $(2)$ to guarantee that the policy of a coalition accepts an exchange if and only if it advantages its members.

\begin{example}\label{ex:1}
Assume \Alice(lice) prefers \ki(iwis) over \lm(emons), while \Bob(ob) preferences are opposite.
Assume also that \Alice\ accepts to perform exchanges with \Bob.
Then, \Alice's policy includes a rule stating that she wills to exchange kiwis for lemons with \Bob.
Say that \Bob\ accepts exchanges with \Alice\ as far as he gets \lm\ in return for \ki\ from someone.
Their policies can be roughly described as the following statements, which permit the direct exchange depicted in~\autoref{fig:ex:1}:
\begin{description}
	\item[\textbf{A1}] I give \Bob\ a \lm\ if I get a \ki\ in return from him;
	\item[\textbf{B1}] I give \Alice\ a \ki\ if I get a \lm\ in return.
\end{description}
\end{example}
Direct exchanges can be composed, e.g. \Alice\ and \Bob\ can exchange two fruits at a time.
However, not every agreement involves two parts only: circular agreements allow agents to perform exchanges that cannot be expressed as combination of direct exchanges.
\begin{example}\label{ex:2}
Assume \Alice\ accepts to exchange a  \ki\ for a \lm, \Bob\ a \ma(andarin) for a \ki, and \Charlie\ a \lm\ for a \ma.
Their policies include:
\begin{description}
	\item[\textbf{A2}] I give any agent a \ki\  if I get a \lm\ in return;
	\item[\textbf{B2}] I give any agent a \ma\ if I get a \ki\  in return;
	\item[\textbf{C1}] I give any agent a \lm\ if I get a \ma\ in return.
\end{description}
Suppose that \Alice, \Bob, \Charlie\ have each a single \ki, \ma, and \lm, and want a \lm, a \ki, and a \ma, respectively.
No exchange between two agents satisfies their policies at the same time.
But a circular exchange as in~\autoref{fig:ex:2} can take place, where each agent gives something to another and is payed by a third one so that all of them are satisfied at the end.
\end{example}
It is not always the case that an agent wants something in return for herself: resources can be given for free (e.g. if their value is less than zero for the owner) or for helping someone else, e.g. a member of the same coalition.
\begin{example}\label{ex:3}
Assume \Alice\ and \Charlie\ are in a coalition: \Alice\ accepts to pay a \ma\ to everyone who gives a \ki\ to \Charlie, and \Charlie\ accepts to give a \lm\ in return for any \ki\ given to \Alice.
The policy of their coalition includes
\begin{description}
	\item[\textbf{AC1}] \Alice\  gives you a \ma\ if you give a \ki\ to \Charlie\  in return;
	\item[\textbf{AC2}] \Charlie\  gives you a \lm\  if you give a \ki\  to \Alice\  in return.
\end{description}
Since \Bob\ accepts to exchange \ki\  for \lm\  (rule \textbf{B1}), the exchange in~\autoref{fig:ex:3} can occur: \Bob\ is happy to receive a \lm\  for a \ki, \Alice\ to receive \ki, and \Charlie\ to pay for her
(in the rules, ``you'' stands for \emph{any} agent).
\end{example}
So far we have only evaluated exchanges in terms of satisfaction of agent's policies.
However, another constraint is put on exchanges: an agent $\usr$ offering a resource $\res$ must possess it before getting the wanted resource. 
In~\autoref{subsec:fair-computations}, we show that $\usr$ can still obtain what she wants, incurring a temporary debt that is paid back by acquiring $\res$ by some other agent.
\begin{example}\label{ex:4}
	Assume that \Bob\ and \Alice\ want to exchange a \lm\ for two \ma\ of \Bob. 
	Their policies then contain the rules:
	\begin{description}
		\item[\textbf{A3}] I give \Bob\ a \lm\ if he gives me two \ma\ in return;
	  	\item[\textbf{B3}] I give \Alice\ two \ma\ if she gives me a \lm\ in return.
	\end{description}
	A direct exchange can take place, if both have the needed resources, but assume that \Alice\ has no \lm\ and \ma.
	\Alice\ can however perform an exchange with \Charlie\  (allowed by rule \textbf{C1}) incurring a temporary debt of \lm, that is given to \Bob\ to obtain two \ma s, and pay the debt (see~\autoref{fig:ex:4}).
\end{example}

\section{Policy Based Exchange Environments}\label{sec:env}

In this section, we first introduce the notion of \emph{exchange environment} as a formal model of scenarios where agents join coalitions and exchange resources according to their preferences and goals. 
This model is a transition system where states are resources allocations and transitions are resource exchanges between agents of a coalition. 
Then, we introduce the notion of \emph{exchange policies} that are statements defined by coalitions in isolation to express what they offer to others and what they want in return.
The transitions must obey the policies of the involved agents and we call them \emph{agreements}.
Moreover, we introduce the notion of \emph{valuation function} to capture the utility that a given resource allocation has for each agent. 
We characterize then a \emph{deal} as an exchange that increases the utility for the agents of a coalition, and a policy as rational when it leads to deals.

\subsection{Exchange Environments}

\noindent
Below, we assume the following finite sets:
a set $\Res$ of \emph{resources}, ranged over by $\res, \res', \res''$ each associated with a (fixed) quantity $q(\res) \in \mathbb{N}$;
a set $\Usr$ of \emph{agents}, ranged over by $\usr, \usr', \usr''$.

We start by defining resource allocations.
Intuitively, they specify the resources each agent owns under the condition that we cannot assign more resources than the available ones. 
Formally,  
\begin{definition}[Resource Allocation]
A \emph{resource allocation} $\MuACstate$ is a function associating each agent with a multiset of $\Res$ such that
$\sum_{\usr \in \Usr} \MuACstate(\usr)(\res) = q(\res)$.
\end{definition}
Next, we introduce the notions of transfer and exchange. 
A \emph{transfer} occurs when a agent $\usr$ sends her resource $\res$ to another agent $\usr'$.
An \emph{exchange} is a finite multiset of transfers.%
\footnote{
		Hereafter we use multisets, i.e. sets with many occurrences of the same object.
		As usual, we represent multisets by 
		functions $f, g, \dots$ from each element to the natural number of its occurrences in the multiset.
		Also, the disjoint union for multisets $(f \uplus g)(x)$ is defined as $f(x) + g(x)$ for every $x$ in the domain.
		For simplicity, we carry the set notation over multisets.
}
Then, we define an \emph{exchange environment} as a transition system with allocations as states and exchanges as transitions.

\begin{definition}[Transfer, Exchange and Exchange Environment] \label{def:ee}
	An \emph{exchange} is a multiset $\exc \in \Exc$ of \emph{transfers} $\tran \in \Tran$, where $\tran$ is a triple 
	$\usr \xmapsto{\res} \usr'$, 
	with $\usr' \neq \usr$.

	An \emph{exchange environment} %
	is a pair $(\MuACStates, \rightarrow)$, where $\rightarrow\ \subseteq \MuACStates \times \Exc \times \MuACStates$ is the \emph{transition relation} that contains the triple $\MuACstate \xrightarrow{\exc} \MuACstate'$ if and only if for all $\usr \in \Usr$ and $\res \in \Res$ the following conditions hold
	\begin{align*}
	\text{(1)} \ &\sum\nolimits_{\usr'} \exc (\usr \xmapsto{\res} \usr') \leq \MuACstate(\usr)(\res) \quad \text{and}\\
	\text{(2)} \ &\MuACstate'(\usr)(\res) = \MuACstate(\usr)(\res) + 
	 \!\sum\nolimits_{\usr''}\!\! \exc (\usr''\!\xmapsto{\res} \usr) - \!\!\sum\nolimits_{\usr'}\!\! \exc (\usr \xmapsto{\res} \usr') 
	\end{align*}
\end{definition}
Condition $(1)$ ensures that an exchange $exc$ is possible only when an agent $a$ owns enough resources. Condition $(2)$ ensures that the allocation is correctly updated and that no resource is created or destroyed.
We say that an agent $\usr$ is \emph{involved} in a transition if it appears in the exchange labeling it.

\subsection{Exchange Policies}
So far, agent intents play no role, and thus there is no guarantee that a transition of the exchange environment complies with their preferences.
We propose an operational characterization of these preferences via \emph{exchange policies} where each policy specifies 
under which conditions an agent accepts an exchange.
The next subsection also introduces 
\emph{valuation functions} to provide a quantitative measure of exchanges, associating each allocation with a utility value for each agent.
As common in real world, agents can join \emph{coalitions}, i.e. sets of agents that define shared policies to obtain mutual benefits.
Coalitions define their policies in isolation, and rely on them to perform decisions about exchanges.

Roughly, an exchange policy is a set of \emph{exchange rules}, written 
$\exc \triangleleft \exc'$ to be read as follows: the agents in the coalition are willing to perform the exchanges $\exc$ in return for the exchange $\exc'$.
Of course, it is not possible to promise transfers on behalf of agents not in the coalition.
The exchange policy determines whether an exchange requires the pay-off to be given directly to the agent who gives away some resources or to another agent in the coalition, therefore allowing some agent to pay for others (and to accept that others are paying for them).
Formally:

\begin{definition}[Exchange Rules and Policies]\label{def:policy} 
	Given a coalition \mbox{$\coal \subseteq \Usr$,} an \emph{exchange rule} is a pair $\exc \triangleleft \exc' \in \Exc \times \Exc$ such that for all 
	$\usr \xmapsto{\res} \usr' \in \exc$ and for all $\usr' \xmapsto{\res} \usr \in \exc'$, $\usr \in \coal$.
	
	The \emph{exchange policy} $pol_\coal$ of $\coal$ is a set of exchange rules. 
\end{definition}
The simplest policies are for single agents (i.e.~singleton coalitions) giving resources for something in return.
	\begin{example}\label{ex:exc-pol-simp}
			We restate the rule \textbf{A1} of the example above, i.e. \Alice\  wants to exchange $\lm$ resources with \Bob\ for $\ki$.
				\begin{align*}
				\pol_{\{\Alice\}} \supseteq \{ \{\Alice\  \xmapsto{\lm} \Bob\} \triangleleft \{ \Bob \xmapsto{\ki} \Alice\  \} \}
				\end{align*}
The rule \textbf{B1} is more general, as \Bob\ does not care who is paying him.
				\begin{align*}
				\pol_{\{\Bob\}} \supseteq \bigcup\nolimits_{\usr \in \Usr} \{ \{\Bob\  \xmapsto{\ki} \Alice\} \triangleleft \{ \usr \xmapsto{\lm} \Bob\  \} \}
				\end{align*}
			Finally, in \textbf{C1}, \Charlie\  accepts to perform exchanges with everybody. 
				\begin{align*}
				\pol_{\{\Charlie\}} \supseteq \bigcup\nolimits_{\usr, \usr' \in \Usr} \{ \{\Charlie\  \xmapsto{\lm} \usr\} \triangleleft \{ \usr' \xmapsto{\ma} \Charlie\  \} \}
				\end{align*}
			The same is also true for the rules \textbf{A2} and \textbf{B2}.
				\begin{align*}
				\pol_{\{\Alice\}} &\supseteq \bigcup\nolimits_{\usr, \usr' \in \Usr} \{ \{\Alice\  \xmapsto{\ki} \usr\} \triangleleft \{ \usr' \xmapsto{\lm} \Alice\  \} \}\\
				\pol_{\{\Bob\}} &\supseteq \bigcup\nolimits_{\usr, \usr' \in \Usr} \{ \{\Bob\  \xmapsto{\ma} \usr\} \triangleleft \{ \usr' \xmapsto{\ki} \Bob\  \} \}				
				\end{align*}
	\end{example}
Note that we can write the rules $\emptyset \triangleleft \exc$ and $\exc' \triangleleft \emptyset$: the first means that the coalition accepts to receive the resources in $\exc$ for free, the second that they are happy to perform the exchange $\exc'$ without receiving anything in return.
	\begin{example}\label{ex:exc-pol}
			In the policy below, \{\Alice, \Charlie\} is a coalition, and \Charlie\ will pay
		with an \lm\ every agent that gives a \ki\ to \Alice\ (rule \textbf{AC2} of~\autoref{ex:2}).
			\begin{align*}
			\pol_{\{\Alice,\Charlie\}} \supseteq \bigcup\nolimits_{a \in A} \{ \{\Charlie \xmapsto{\lm} a\} \triangleleft \{ a \xmapsto{\ki} \Alice \} \}
			\end{align*}
	\end{example}

We now move towards the definition of agreements as exchanges that satisfy the policies of all agents involved.
First, we say that an exchange is a accepted by a coalition when it respects its policy.
Intuitively, this happens if all the agents of the coalition receive in return (as a payoff) what they are asking for each resource that they are giving away (as a contribution).
Note that this check can be done by the agents of the coalition in isolation.
Formally:

\begin{definition}[Accepted Exchanges]\label{def:accepted-exc} 
	Let $\pol_\coal \vDash_{\exc'} \exc$ be the smallest relation over $\Pol \times \Exc \times \Exc$ such that 
	\begin{enumerate}
		\item $\pol_{\coal} \vDash \emptyset \triangleleft \emptyset$;
		\item $\pol_\coal \vDash \exc \triangleleft \exc'$, if $\exc \triangleleft \exc' \in \pol_{\coal}$;
		\item $\pol_\coal \vDash (\exc_1 \uplus \exc_2) \triangleleft (\exc_1' \uplus \exc_2')$, 
			if $\pol_\coal \vDash \exc_i \triangleleft \exc_i'$, $i = 1,2$.
	\end{enumerate}
The coalition $\coal$ accepts the exchange $\exc \uplus \exc'$ if $\pol_\coal \vDash \exc' \triangleleft \exc$, and we call $\exc$ its contribution and $\exc'$ its payoff.
\end{definition}

Intuitively, an exchange is an \emph{agreement} if it satisfies the policies of all the involved coalitions.
Consider~\autoref{ex:1}: the exchange $\exc = \{\Alice\xmapsto{\lm}\Bob, \Bob\xmapsto{\ki}\Alice\}$ is an agreement because the left part of the rule of each agent matches the right part of the other, and their union is $\exc$.
This condition is lifted up to sets of rules: the union of the left parts of some rules of the agents must equate the union of all their right parts.
An example is the circular agreement $\{\Alice\xmapsto{\ki}\Bob, \Bob\xmapsto{\ma}\Charlie, \Charlie\xmapsto{\lm}\Alice\}$ in~\autoref{fig:ex:3}, which is obtained by using the rules \textbf{A2}, \textbf{B2} and \textbf{C1} of~\autoref{ex:exc-pol-simp}.
\begin{definition}[Agreement]\label{def:fairexc} 
	An exchange $\exc$ is an \emph{agreement} if and only if for all coalition $\coal$ such that $\coal \subseteq \Usr$ there exists a pair of exchanges $\exc_{\coal}$ and $\exc_{\coal}'$ such that $\pol_{\coal} \vDash {\exc_\coal}' \triangleleft \exc_{\coal}$ and $\biguplus_{\coal \subseteq \Usr} \exc_{\coal}' = \biguplus_{\coal \subseteq \Usr} \exc_{\coal} = \exc$.

\end{definition}
Note that the disjoint union of agreements is still an agreement.
Actually, disjoint union is critical for defining agreements.
Verifying disjointness requires a sort of global check on the partitioning of the exchanged resources.
Otherwise, the same resource can be \emph{offered more than once} to different agents (a sort of double spending), and this may go unnoticed since each coalition
only knows its policy and its portion of the overall exchange. 
\begin{example}\label{ex:fair-nonlocal}
	Consider~\autoref{ex:exc-pol-simp} and \ref{ex:exc-pol}, and the following exchanges where 
	\Bob\  gives an \ki\ to \Alice\ and both \Alice\ and \Charlie\ pay for it with an \lm\ (i.e. a double spending occurs):
	\begin{align*}
		\exc = \{ \Bob \xmapsto{\ki} \Alice, \Alice \xmapsto{\lm} \Bob, \Charlie \xmapsto{\lm} \Bob \}. 
	\end{align*}
	The exchange may seem fair to both \Alice\ (for \textbf{A1}) and \Charlie\ (for \textbf{AC2}), if $\exc$ is 
	decomposed as the (non-disjoint) union of 
		\begin{align*}
			\exc_{\{\Alice,\Charlie\}} &= \{ \Charlie \xmapsto{\lm} \Bob \} &
			             \exc_{\{\Alice,\Charlie\}}' &= \{ \Bob \xmapsto{\ki} \Alice \} \\
			\exc_{\{\Alice\}} &= \{ \Alice \xmapsto{\lm} \Bob \} & \exc_{\{\Alice\}}' &= \{ \Bob \xmapsto{\ki} \Alice \} \\
			\exc_{\{\Bob\}} &= \{ \Bob \xmapsto{\ki} \Alice, \Bob \xmapsto{\ki} \Alice \} &
			\exc_{\{\Bob\}}' &= \{ \Alice \xmapsto{\lm} \Bob, \Charlie \xmapsto{\lm} \Bob \}
		\end{align*}
	Of course, \Bob\ is scamming \Alice\ and \Charlie: they are both paying for the same resource.
	Indeed $\exc$ is not an agreement, but it is $\exc' = \{ \Bob \xmapsto{\ki} \Alice$, $\Bob \xmapsto{\ki} \Alice, \Alice \xmapsto{\lm} \Bob, \Charlie \xmapsto{\lm} \Bob \} = \exc_{\{\Alice,\Charlie\}} \uplus \exc_{\{\Alice\}} \uplus \exc_{\{\Bob\}} = \exc_{\{\Alice,\Charlie\}}' \uplus \exc_{\{\Alice\}}' \uplus \exc_{\{\Bob\}}'$ (note that $\exc'$ is the disjoint union of the exchanges in~\autoref{fig:ex:1} and \ref{fig:ex:3}).
\end{example}

\subsection{Policies and Valuation Functions}
We now characterize when a coalition increases its evaluation of an assignment through an agreement.
First, each resource is assigned a value by each agent $\usr$, using which $\usr$ gets the overall value of an assignment.
Note that the following definition introduces special valuation functions that are called additive in~\cite{Porello2010}.

\begin{definition}
Let $av_{\usr}: \Usr \times \Res \rightarrow \mathbb{Z}$ represent the (positive or negative) value that the agent $\usr$ 
associates with the case that $\usr'$ holds a resource $\res$.
\\
The \emph{valuation function} $v_{\usr} \colon \MuACStates \rightarrow \mathbb{Z}$ of the agent $\usr$ is the function 
$v_{\usr}(\MuACstate) = \sum_{\usr' \in \Usr} \sum_{\res \in \Res} \MuACstate(\usr')(\res) \cdot av_{\usr}(\usr',\res)$.
\end{definition}

\begin{definition}[Deal]
Given a coalition $\coal$ and an exchange $\exc$, let $\exc\!\downarrow_\coal = \{\usr_1 \xmapsto{r} \usr_2 \in \exc \mid \usr_1 \in \coal \text{ or } \usr_2 \in \coal\}$.
\\
Then ${\exc\!\downarrow_\coal}$ is a \emph{deal} for  $\coal$ if and only if $\, \forall \usr \in \coal$ all the transitions $\MuACstate \xrightarrow{\exc\!\downarrow_\coal} \MuACstate'$ 
are such that $v_{\usr}(\MuACstate') \geq v_{\usr}(\MuACstate)$.
\end{definition}

\begin{restatable}[Policies and deals]{theorem}{polanddeal}\label{th:pol-deal}
Let $\coal$ be a coalition and let $\exc$ be an exchange accepted by $\pol_{\coal}$.
If every $\exc'' \triangleleft \exc' \in \pol_{\coal}$ is such that $\,\forall \usr \in \coal.\,W(a, \exc'' \uplus \exc') \geq 0$,
then $\exc$ is a deal for $\coal$, 
where $W(\usr, \{\usr_i \xmapsto{\res_i} \usr_i' \}_{i \in I}) = \sum_{i \in I} 
(av_{\usr}(a_i', r_i) - av_{\usr}(a_i, r_i))$.
\end{restatable}

An immediate consequence of \autoref{th:pol-deal} is that a policy is \emph{sound}, 
in that it accepts only deals for $\coal$ if its rules $\exc' \triangleleft \exc$ are such that
the increment of value granted by its payoff $\exc'$ is greater than the loss of its contribution $\exc$.
One would like to have in addition that the policy $\pol_{\coal}$ accepts all the deals for $\coal$, a sort of \emph{completeness}.
If both cases hold, one has the so-called \emph{rational} policies.%

While \autoref{th:pol-deal} suggests an easy way of checking a policy sound,  verifting it is complete may require a brute force analysis in the case of finite resources.
Also, one can build one (out of the equivalent) less restrictive correct policy, namely a rational one, starting from the valuation function in hand, because the proof of the following theorem is constructive.
The idea is to split an $exc\!\downarrow_C$ in all possible pairs $\exc'', \exc' $
such that their weight $W$ is non negative and insert the rule $\exc'' \triangleleft \exc'$ in the policy $\pol_C$.
\begin{restatable}{theorem}{existsrational}\label{th:exists-rational}
Given the valuation functions of all the agents of a coalition $\coal$, there exists a rational policy for $\coal$.
\end{restatable}
Even when the policies of all the coalitions are rational, it is not always the case that a transition is a deal for \emph{every} coalition.
E.g. let $\coal = \{\usr_1, \usr_2\}$ and $\coal' = \{ \usr_1', \usr_2' \}$ be two disjoint coalitions and let $\usr_1$ and $\usr_1'$ exchange some resources with $\usr_2$ and $\usr_2'$, respectively.
Even if the gain of $\usr_1$ performing the exchange with $\usr_2$ is positive for $\coal$, it may be less than the loss of value for $\usr_1$ caused by $\usr_1'$ and $\usr_2'$ performing their exchange.
This is never the case with special valuation functions, through which an agent assigns the same value to all the resources that does not belong to her.
In this case an agreement, 
which is also a deal, causes a quasi-Pareto improvement, because the value of the allocation does not decrease for all agents.
\begin{restatable}{theorem}{pareto}\label{th:Pareto}
If every policy is rational and $av_{\usr}(\usr',\res) = k_a$ for some fixed $k_a \in \mathbb{Z}$ when $\usr \neq \usr'$, then each transition 
$\st \xrightarrow{\exc} \st'$ with $\exc$ an agreement is such that $v_\usr(\st') \geq v_\usr(\st)$ for all $\usr$.
\end{restatable}

\section{A Logic for Characterizing Agreements}\label{sec:col-formal}

\begin{figure*}[t]
	\begin{center}
		\small
		\begin{tabular}{c}		
			\prftree[r]
			{(Cont)}
			{!\phi, !\phi, \Phi \vdash \phi'}
			{!\phi, \Phi \vdash \phi'}
			
			\qquad
			
			\prftree[r]
			{(Weak)}
			{\Phi \vdash \phi'}
			{!\phi, \Phi \vdash \phi'}

			\qquad
			
			\prftree[r]
			{($!$-left)}
			{\phi, \Phi \vdash \phi'}
			{!\phi, \Phi \vdash \phi'}
			
			\qquad
		
			\prftree[r]
			{($I$-left)}
			{\Phi \vdash \phi}
			{I, \Phi \vdash \phi}
			
			\qquad

			\prftree[r]
			{($I$-right)}
			{}
			{\vdash I}
			
			\\[.3cm]
			
			\prftree[r]
			{(Ax)}
			{\phi \vdash \phi}
			
			\qquad
			
			\prftree[r]
			{($\otimes$-left)}
			{\phi, \phi', \Phi \vdash \phi''}
			{\phi \otimes \phi', \Phi \vdash \phi''}
			
			\qquad
			
			\prftree[r]
			{($\otimes$-right)}
			{\Phi \vdash \phi\quad}
			{\Phi' \vdash \phi'}
			{\Phi, \Phi' \vdash \phi \otimes \phi'}
			
			\\[.3cm]

			\prftree[r]
			{($\multimap$-left)}
			{\Sigma \vdash \sigma}	
			{\Sigma, \sigma \multimap \sigma' \vdash \sigma'}
			
			\qquad
			
			\prftree[r]
			{($\linearcontract$-left)}
			{\delta \vdash \delta'\quad}
			{\delta' \vdash \delta\quad}
			{\Phi, \delta' \vdash \sigma}
			{\Phi, \delta \linearcontract \delta' \vdash \sigma}
			
			\qquad
			
			\prftree[r]
			{($\linearcontract$-split)}
			{\Phi, \delta \otimes \delta'' \linearcontract \delta' \otimes \delta''' \vdash \sigma}
			{\Phi, \delta \linearcontract \delta', \delta'' \linearcontract \delta''' \vdash \sigma}
		\end{tabular}
	\end{center}
	\caption{\MuACL\ rules.}
	\vspace{\baselineskip}
	\label{fig:MuACLfull}
\end{figure*}

So far, we have characterised agreements at the basic level of exchange environments and coalition policies.
We develop now a linear logic for modeling exchange environments, policies and agreements, and reduce the check of an exchange being an agreement to the validity of the sequent modelling it.

We choose linear logic for its unique features to declaratively represent resources and their usage.
In this view, resources are represented as logical assumptions in a proof, and each of them can only be used once and only once during the proof: resources can neither be duplicated nor thrown away at will.
A deduction in the logic models the way resources are manipulated, and this is convenient in our approach to formalize exchange environments and their behaviour.

\subsection{Contractual Exchange Logic}\label{sec:col-logic}

To define Contractual Exchange Logic (\MuACL) we start from a computational fragment of linear logic, following~\cite{Kanovich94}, and we then extend it with a new operator, inspired by PCL~\cite{BZ}, to express the 
typical offer/return actions of contracts. 
\begin{definition}[\MuACL\ Propositions]
	A \MuACL\ proposition $\phi$ is defined as
	\begin{align*}
		\phi &::= \sigma \mid \delta \mid \theta \mid \omega \mid \xi &
		\sigma &::= I \mid \res@\usr \mid \sigma \otimes \sigma\\
		\delta &::= I \mid \res@\usr \multimap \res@\usr \mid \delta \otimes \delta &
		\theta &::= I \mid \delta \linearcontract \delta \mid \theta \otimes \theta\\
		\xi &::=\ !\delta \mid \xi \otimes \xi &
		\omega &::=\ !\theta \mid\ \omega \otimes \omega 
	\end{align*}
	We denote multisets of propositions using (the corresponding) Greek capital letters: $\Phi, \Sigma, \Delta, \Theta, \Xi, \Omega$.
\end{definition}
We refer to the common resource-based interpretation of linear logic for describing the intuitive meaning of the propositions above.
In this interpretation $\res@\usr$ stands for a resource association, meaning that $\res$ currently belongs to the agent $\usr$.
A predicate $\res@\usr \multimap \res'@\usr'$ is a consumable processes (they can run only once) transforming the atomic $\res@\usr$ into $\res'@\usr'$.
Predicates of the form $\delta \linearcontract \delta'$, composed with our new operator called \emph{linear contractual implication}, are promises stating that $\delta'$ will be performed if also $\delta$ is.
Finally, $\omega$ and $\xi$ represents (non-linear) information about promises and processes that can be used ad libitum, and tensor product allows composing multisets of previous entities, where $I$ (representing \emph{true} in linear logic) is the empty multiset.
The \MuACL\ sequents are as follows.
\begin{definition}[\MuACL\ Sequent]
A \MuACL\ \emph{sequent} is of form
\begin{align*}
\Omega, \Xi, \Theta, \Delta, \Sigma \vdash \phi.
\end{align*}
A sequent is \emph{initial} if $\Theta, \Delta = \emptyset$ and $\phi = \sigma$ for some $\sigma$, i.e. if it has the form $\Omega, \Xi, \Sigma \vdash \sigma$
(we omit hereafter the empty components).
\end{definition}
The \MuACL\ sequent $\Xi, \Omega, \Theta, \Delta, \Sigma \vdash \sigma$
intuitively means that the state $\sigma$ is a possible transformation of $\Sigma$ using the processes and promises in the assumptions $\Xi, \Omega, \Theta, \Delta$.

The deduction system for $\vdash$ are in~\autoref{fig:MuACLfull}.
They mostly result from instantiating the standard ones of the multiplicative fragment of linear logic on the \MuACL\ sequents.
Note that we omit the cut rule in this fragment.
In addition there are two rules for the linear contractual implication: the ($\linearcontract$-left) rule introduces the operator on the left if what is required by the contract is satisfied by the consequences; the
($\linearcontract$-split) rule deals with composition of contracts.

\begin{example}
		A linear implication $\res@\usr \multimap \res@\usr'$ represents a transfer where a predicate $\res@\usr$ is consumed and a new $\res@\usr'$ is created.
		Note that $\res@\usr \multimap \res@\usr',$ $\res@\usr \vdash \res@\usr'$ is indeed valid.
		
		A linear contractual implication $\delta \linearcontract \delta'$ encodes a promise of $\delta'$ in return of $\delta$.
		Direct exchanges, that will be used to encode agreements like the one of~\autoref{fig:ex:1}, are expressed by a sequent of the form
			\[
			\delta \linearcontract \delta', \delta' \linearcontract \delta, \Sigma \vdash \sigma\
			\]
			where the exchange $\delta'$ is promised in return for $\delta$ and vice versa.
			The following derivation proves that the exchange $\delta, \delta'$ that transforms the state $\Sigma$ in $\sigma$ can be performed (we omit the proofs $\Pi$ and $\Pi'$, using (Ax), ($\otimes$-left) and ($\otimes$-right) rules, as they are straightforward):
			\begin{align*}
				\prftree[r]
				{($\linearcontract$-split)}
				{
					\prftree[r]
					{($\linearcontract$-left)}
					{
						\prftree
						{\Pi}
						{\delta \otimes \delta' \vdash \delta' \otimes \delta}
					}
					{
						\prftree
						{\Pi'}
						{\delta' \otimes \delta \vdash \delta \otimes \delta'}
					}
					{
						\prftree[r]
						{($\otimes$-left)}
						{
							{\delta, \delta', \Sigma \vdash \sigma}
						}
						{\delta' \otimes \delta, \Sigma \vdash \sigma}
					}
					{\delta \otimes \delta' \linearcontract \delta' \otimes \delta, \Sigma \vdash \sigma}	
				}
				{\delta \linearcontract \delta', \delta' \linearcontract \delta, \Sigma \vdash \sigma}	
			\end{align*}
			For the circular exchange of~\autoref{fig:ex:2}, we use \\
			$\delta \linearcontract \delta', \delta' \linearcontract \delta'', \delta'' \linearcontract \delta, \Sigma \vdash \sigma$, and apply ($\linearcontract$-split) twice.
\end{example}

Noticeably, \MuACL\ proofs can be normalized\footnote{
In the following, we will call \emph{proof} the derivation of a theorem from the axioms, and only use the term \emph{derivation} for a derivation with open assumptions, i.e. a proof tree where the leaves are not only axioms.
We also say that two proofs are \emph{equivalent} if they prove the same sequent.
}.

\begin{definition}[Normal Proofs]\label{def:normalforms}
	A \MuACL\ proof for an initial sequent is \emph{normal} if 
	it can be decomposed in either form, where 
	$\Pi_1$ only uses (Weak), (Cont), ($\otimes$-left) and ($!$-left) rules;
	$\Pi_2$ only ($\linearcontract$-split); 
	$\Pi_3$ and $\Pi_3'$ only ($\otimes$-right), ($\otimes$-left), (Ax), ($I$-right) and ($I$-left); 
	$\Pi_4$ only ($\multimap$-left), ($\otimes$-right), ($\otimes$-left), (Ax), ($I$-right) and ($I$-left).
	\begin{center}
		\begin{tabular}{c}
			\begin{tabular}{c}
				\prfsummary[$\Pi_1$]
				{
						\prftree
						{\Pi_4}
						{\Delta, \Sigma \vdash \sigma}
				}
				{\Omega, \Xi, \Sigma \vdash \sigma}\\[0.2cm]
				\textit{normal form 1}
			\end{tabular}
			
			\begin{tabular}{c}
				\hspace{0.5cm}	
				\prfsummary[$\Pi_1$]
				{
						\prfsummary[$\Pi_2$]
						{
							\prftree[r]
							{($\linearcontract$-left)}
							{\Pi_3}
							{\Pi_3'}
							{
								\prftree
								{\Pi_4}
								{\Delta, \delta, \Sigma \vdash \sigma}
							}
							{\theta, \Delta, \Sigma \vdash \sigma}
						}
						{
							{\Theta, \Delta, \Sigma \vdash \sigma}
						}
				}
				{\Omega, \Xi, \Sigma \vdash \sigma}\\[0.2cm]
				\textit{normal form 2}
			\end{tabular}
		\end{tabular}
	\end{center}
\end{definition}
These two normal forms are general: a proof exists for an initial sequent only if a proof in normal form exists.
\begin{restatable}[\MuACL\ Normal Proofs]{theorem}{normalform}\label{thm:col-normal-ncr}
	For any $\Omega, \Sigma, \sigma$, the initial sequent $\Omega; \Sigma \vdash \sigma$ is valid in 
	\MuACL\ if and only if a normal proof\  $\,\Pi$ exists for $\Omega; \Sigma \vdash \sigma$.
\end{restatable}

For reasoning about exchange environments a logic must be decidable.
We now prove that \MuACL\  is such and for that it is enough to consider normal proofs only.
Note that the normal form 1 corresponds to the case where no contractual rule is applied, hence we can assume $\Omega = \emptyset$.
We are thus in the context of (a fragment of) standard linear logic, and decidability follows from a suitable application of Kanovich's technique~\cite{Kanovich94} that reduces validity to reachability in Petri Nets, which can be decided using~\cite{PetriReach}.
\begin{restatable}{lemma}{fstnfdecidencr}\label{thm:fstnfdecide-ncr}
	An always-terminating algorithm exists that decides the existence of a proof in the normal form 1 for a given initial sequent.
\end{restatable}

Finally, for the normal form 2 we reduce to the previous case.
We prove that a proof in the normal form 2 can be effectively rewritten in the normal form 1.
The reduction is performed in an algebraic framework
by considering the derivations in a bottom-up fashion, starting with the sequent we are proving and constructing the premises.

Consider a semiring module\footnote{A semiring module is a generalization of the notion of vector space in which the field of scalars is replaced by a semiring.} $M$ over the set of natural numbers $\mathbb{N}$ with subformulas of any \MuACL\ predicate $\phi$ as its basis (we can safely reduce to the finite set of the ones appearing in the considered sequent).
For simplicity, we call the elements $\bar x$ of $M$ vectors, and write $\bar x(\phi)$ for the number associated by $\bar x$ to the basis element $\phi$.

Roughly, we show that valid premises $\Delta, \Sigma \vdash \sigma$ for $\Pi_1$ in normal form 1 correspond to the linear combinations of a finite set of vectors 
that depends on the $\delta$ predicates appearing in the $\Xi$ of the sequent that we are proving.
This intuitively encodes the fact that we can take $\delta$ formulas ad libitum, because of the $!$ operator.
We then replicate a similar construction for normal form 2.
In particular, the existence of $\Pi_1$ depends on the encoding of the premises being a linear combinations of the encoding of $\theta$ and $\delta$ appearing in $\Omega$ and $\Xi$ of the conclusions;
the premises of $\Pi_2$ are uniquely determined by a linear function, and $\Pi_3$ and $\Pi_3$ correspond to checking an homogeneous system of linear equations.

Then, by the Hilbert basis theorem~\cite{HB}, we can combine these conditions and represent the set of solutions as the linear combinations of a finite set of vectors representing $\delta$ predicates, i.e. a $\Pi_1$ derivation for a proof in normal form 1. 

\begin{restatable}{lemma}{sftoff}\label{thm:col-sftoff}
	For any $\Omega, \Xi, \Delta, \Sigma, \sigma$, there is a computable multiset $\Xi'$ such that there exists a derivation in the normal form 2 from the sequent $\Delta, \Sigma \vdash \sigma$ to $\Omega, \Xi, \Sigma \vdash \sigma$ if and only if there exists a derivation in the normal form 1 from $\Delta, \Sigma \vdash \sigma$ to $\Xi', \Sigma \vdash \sigma$.
\end{restatable}
The following theorem proves that \MuACL\ is decidable.
Its statement mentions initial sequents only, which are however sufficient to reason about agreement transitions, as \autoref{th:fair-exchange} below will make clear.

\begin{restatable}[\MuACL\ Decidability]{theorem}{MuACLsdec}\label{thm:MuACLdecide1}
	An always-terminating algorithm exists that decides if an initial sequent is valid in \MuACL.
\end{restatable}

\subsection{Deriving Transitions of Exchange Environments}

In the following we show how to encode exchange environments and policies as \MuACL\ propositions.

\begin{figure*}[t]
\small
\begin{align*}
	\prfsummary[$\Pi$]
	{
		\prfsummary[$\Pi'$]
		{
			\prftree[r]{($\linearcontract$-left)}
			{
				\prftree[r]{(Ax)}
				{}
				{\semantics{\exc} \vdash \semantics{\exc}}
			}
			{
				\prftree[r]{(Ax)}
				{}
				{\semantics{\exc} \vdash \semantics{\exc}}
			}
			{
				\prftree
				{\Pi''}
				{\semantics{\exc}, \semantics{\Sigma} \vdash \ki@\Bob \otimes \ma@\Charlie \otimes \lm@\Alice}
			}
			{
			\semantics{\exc} \linearcontract \semantics{\exc}, \semantics{\MuACstate} \vdash \ki@\Bob \otimes \ma@\Charlie \otimes \lm@\Alice
			}
		}
		{
		(\lm@\Charlie \multimap \lm@\Alice) \linearcontract (\ki@\Alice\ \multimap \ki@\Bob),
		(\ki@\Alice \multimap \ki@\Bob) \linearcontract (\ma@\Bob\ \multimap \ma@\Charlie),
		(\ma@\Bob \multimap \ma@\Charlie) \linearcontract (\lm@\Charlie\ \multimap \lm@\Alice),
		\semantics{\MuACstate} \vdash \ki@\Bob \otimes \ma@\Charlie \otimes \lm@\Alice
		}
	}
	{\semantics{\pol_{\{\Alice\}}}, \semantics{\pol_{\{\Bob\}}}, \semantics{\pol_{\{\Charlie\}}}, \semantics{\MuACstate} \vdash \ki@\Bob \otimes \ma@\Charlie \otimes \lm@\Alice}	
\end{align*}
\caption{\MuACL\ proof in the normal form 2, the proof $\Pi''$ is omitted for simplicity, and uses (Ax), ($\otimes$-left), ($\otimes$-right) and ($\multimap$-left).}
\label{fig:ex:enc}
\end{figure*}

\begin{definition}\label{def:compile}
We write $\phi^n$ with $n \in \mathbb{N}$ for the tensor product of $n >0$ instances of the proposition $\phi$, meaning $I$ if $n = 0$.
	\begin{align*}
	\semantics{\MuACstate} &::= \bigotimes\nolimits_{\usr \in \Usr, \res \in \Res} (\res@\usr)^{\MuACstate(\usr)(\res)}\\
	\semantics{\exc} &::= \bigotimes\nolimits_{\usr,\usr' \in \Usr, \res \in \Res} (\res@\usr \multimap \res@\usr')^{\exc(\usr \xmapsto{\res} \usr')}\\
	\semantics{\pol_\coal} &::= \bigotimes\nolimits_{\exc \triangleleft \exc' \in \pol_\coal} !(\semantics{\exc'} \linearcontract \semantics{\exc})
	\end{align*}
\end{definition}

\begin{example}\label{ex:enc}
Consider a state $\MuACstate$ with three agents \Alice, \Bob, and \Charlie\ owning a \ki, a \ma\ and an \lm, respectively. 
Then its encoding is as follows:
\[
\semantics{\MuACstate} = \ki@\Alice \otimes \ma@\Bob \otimes \lm@\Charlie
\]

The encoding of $\exc = \{ \Alice \xmapsto{\ki} \Bob, \Bob \xmapsto{\ma} \Charlie, \Charlie \xmapsto{\lm} \Alice \}$ is 
\[
\semantics{\exc} = (\ki@\Alice \multimap \ki@\Bob) \otimes (\ma@\Bob \multimap \ma@\Charlie) \otimes (\lm@\Charlie \multimap \lm@\Alice)
\]

Finally, suppose that the policies $\pol_{\{\Alice\}}, \pol_{\{\Bob\}}, \pol_{\{\Charlie\}}$ contain only the rules \textbf{A2}, \textbf{B2} and \textbf{C1} in~\autoref{ex:exc-pol-simp}, respectively.
Their encoding follows.
\begin{align*}
\semantics{\pol_{\{\Alice\}}} =  
&\bigotimes\nolimits_{a, a' \in \{\Bob, \Charlie \}} !((\lm@a \multimap \lm@\Alice) \linearcontract (\ki@\Alice\ \multimap \ki@a')) \\
\semantics{\pol_{\{\Bob\}}} = &\bigotimes\nolimits_{a, a' \in \{\Alice, \Charlie \}} !((\ki@a \multimap \ki@\Bob) \linearcontract (\ma@\Bob\ \multimap \ma@a')) \\
\semantics{\pol_{\{\Charlie\}}} = &\bigotimes\nolimits_{a, a' \in \{\Alice, \Bob \}} !((\ma@a \multimap \ma@\Charlie) \linearcontract (\lm@\Charlie\ \multimap \lm@a')) 
\end{align*}

The following sequent is valid, as shown in the proof in~\autoref{fig:ex:enc}
\[
\semantics{\pol_{\{\Alice\}}}, \semantics{\pol_{\{\Bob\}}}, \semantics{\pol_{\{\Charlie\}}}, \semantics{\MuACstate} \vdash \ki@\Bob \otimes \ma@\Charlie \otimes \lm@\Alice
\]
Note that $\exc$ is an agreement labelling the transition $\MuACstate \xrightarrow{\exc} \MuACstate'$ where $\semantics{\MuACstate'} = \ki@\Bob \otimes \ma@\Charlie \otimes \lm@\Alice$.
\end{example}

In the theorem below the current allocation $\MuACstate$ and the policies $\pol_{\coal}$ determine the left part of an initial sequent, while the right part is for the candidate next allocation $\MuACstate'$.  
Then, a transition $\MuACstate \xrightarrow{\exc} \MuACstate'$ where $\exc$ is an agreement exists if and only if the obtained initial sequent is valid.
Note that this result implies that \MuACL\ proofs are witnesses for fairness of exchanges.

\begin{restatable}{theorem}{correction}\label{th:fair-exchange}
Let  $(\MuACStates, \rightarrow)$ be an exchange environment; let $\MuACstate, \MuACstate' \in \MuACStates$; and  
let $\pol_{\coal}$ be the policies of the coalition $\coal$. 

Then $\biguplus_{\coal \in 2^{\Usr}} \semden{\pol_\coal}, \semden{\MuACstate}\! \vdash\! \semden{\MuACstate'}$ is valid if and only if there exists an agreement $\exc$ such that $\MuACstate\! \xrightarrow{\exc}\! \MuACstate'$.

\end{restatable}

The proof is carried on in three steps. 
First, we note that a proof for $\biguplus_{\coal \in 2^{\Usr}} \semden{\pol_\coal}, \semden{\MuACstate} \vdash \semden{\MuACstate'}$ can always be transformed in one in the normal form 2, where $\Xi = \Delta = \emptyset$.
Then we show that the derivations $\Pi_1$, $\Pi_2$ and $\Pi_3$ exist whenever $\delta = \semden{\exc}$ for some $\exc$ that is an agreement.
Finally, a proof $\Pi_4$ exists for $\semden{\exc}, \semden{\MuACstate} \vdash \semden{\MuACstate'}$ if and only if 
$\MuACstate \xrightarrow{\exc} \MuACstate'$ is a transition of the exchange environment.

\section{Extending \MuACL\ with debts}\label{subsec:fair-computations}

So far, we have only considered exchanges where no debts are permitted: an agent must possess a resource she promises as required by condition $(1)$ of~\autoref{def:ee}.
In case agents trust each other or there is a regulating trusted third party, it is possible to extend our model and logic to consider a wider class of transitions, as that in~\autoref{ex:4}.
This requires updating the exchange environments by weakening the condition $(1)$, allowing an agent to incur a temporary debt.

\begin{definition}
An \emph{exchange environment with debts} is $(\MuACStates, \dashrightarrow)$ with 
$\MuACStates$ defined as in~\autoref{def:ee}, and $\dashrightarrow\ \subseteq \MuACStates \times \Exc \times \MuACStates$ contains the triples $\MuACstate \xdashrightarrow{\exc} \MuACstate'$ if and only if for all $\usr \in \Usr$ and $\res \in \Res$ both (2) from~\autoref{def:ee} and the following hold
\[
\text{(1$'$)}\  \sum\nolimits_{\usr'} \exc (\usr \xmapsto{\res} \usr') \leq \MuACstate(\usr)(\res) + \sum\nolimits_{\usr'} \exc (\usr' \xmapsto{\res} \usr)
\]
\end{definition}
For brevity, we write below $\rightarrow_{ok} \subseteq \MuACStates \times \MuACStates$ for the transition induced by agreements only: $\MuACstate \rightarrow_{ok} \MuACstate$ if and only if $\MuACstate \xrightarrow{\exc} \MuACstate$ for some agreement $\exc$.
As a matter of fact, in an exchange environment, every pair of allocations are connected by a transition, but not by one labelled by an agreement.
Now, we similarly filter $\dashrightarrow$ and define $\dashrightarrow_{ok}$ to be the transition relation in an exchange environment with debits for which a fair exchange exists.
Note that some allocations reachable with these transitions cannot be reached without permitting debits, i.e. $\rightarrow_{ok} \subsetneq \dashrightarrow_{\!ok}$.

\begin{example}
Consider again~\autoref{ex:4}, where \Bob\ and \Alice\ want to exchange a \lm\ for two \ma\ of \Bob. 
Assume that the current assignment $\MuACstate$ is such that \Alice\ has nothing, \Bob\ has two \ma\ and \Charlie\ has one \lm.

The relevant policy rules are \textbf{A3}, \textbf{B3}, \textbf{C1}:
\begin{align*}
	\pol_{\{\Alice\}} &\supseteq  \{ \{\Alice\  \xmapsto{\lm} \Bob\} \triangleleft \{ \Bob\ \xmapsto{\ma} \Alice\  \} \}\\
	\pol_{\{\Bob\}} &\supseteq \{ \{\Bob\  \xmapsto{\ma} \Alice, \Bob\  \xmapsto{\ma} \Alice\} \triangleleft \{ \Alice\ \xmapsto{\lm} \Bob\  \} \}\\
	\pol_{\{\Charlie\}} &\supseteq \bigcup\nolimits_{\usr, \usr' \in \Usr} \{ \{\Charlie\  \xmapsto{\lm} \usr\} \triangleleft \{ \usr' \xmapsto{\ma} \Charlie\  \} \}
\end{align*}
The exchange $\exc = \{ \Charlie \xmapsto{\lm} \Alice,  \Alice  \xmapsto{\lm} \Bob, \Bob  \xmapsto{\ma} \Alice, \Bob  \xmapsto{\ma} \Alice\}$ depicted in  \autoref{fig:ex:4} is forbidden in $\MuACstate$ using exchange environments without debts, since \Alice\ has no \lm s.
Instead, the transition $\MuACstate \xdashrightarrow{\exc} \MuACstate'$ results in the state $\MuACstate'$ where \Alice\ and \Charlie\ have each a \ma, and \Bob\ a \lm. 
Note that the transfer $\Alice  \xmapsto{\lm} \Bob$ causes a temporary debt of \Alice, which is repaid with the transfer $\Charlie \xmapsto{\lm} \Alice$.
\end{example}

Again logic comes to our rescue for deciding if a transition in an exchange environment with debts is an agreement.
This is done by adding the rule (Cut) in \autoref{fig:starcut}.
Consequently, we extend the correspondence between exchange environments and \MuACL\ in presence of debts through the following corollary of~\autoref{th:fair-exchange}.
\begin{figure}[t]
	\begin{gather*}
	\prftree[r]
	{(Cut)}
	{\Phi \vdash \sigma}
	{\Phi', \sigma \vdash \phi}
	{\Phi, \Phi' \vdash \phi}	
	\end{gather*}
	\caption{Cut rule for \MuACL.}
	\label{fig:starcut}
	\vspace{\baselineskip}	
	\medskip
\end{figure}

\begin{restatable}{corollary}{Logiccorrectnessstar}\label{thm:correctcompletestar}
	Under the same conditions~of \autoref{th:fair-exchange}, a transition $\MuACstate \dashrightarrow_{\!ok} \MuACstate'$ exists, if and only if $\,\biguplus_{\coal \in 2^{\Usr}} \semden{\pol_\coal}, \semden{\MuACstate} \vdash \semden{\MuACstate'}$ is valid in \MuACL\ augmented with the (Cut) rule.
\end{restatable}

Decidability of \MuACL\ is not affected by the (Cut) rule.
\begin{restatable}{corollary}{MuACLsdec}\label{thm:MuACLdecide}
	An always-terminating algorithm exists that decides if an initial sequent is valid in \MuACL\ augmented with the (Cut) rule.
\end{restatable}

\section{Related Work}\label{sec:col-related}

The problem of fairly exchanging electronic assets among a set of agents has been addressed by different communities, e.g. artificial intelligence, fair exchange protocols in distributed systems~\cite{Bao99,Even,Franklin98,Freiling}.
Below, we focus on those approaches that use linear logic to capture these issues, and we conclude with a comparison with logics having contractual aspects.

\paragraph*{Logical Modelling of Resource Exchange}
 
Linear logic has been used to model resource-aware games and problems in the artificial intelligence community.
They all describe the desire of agents in terms of their goals or valuation functions, and derive or recognise reasonable offers and strategies.
A contribution of ours is instead a way of directly modelling what agents offer via exchange policies, combining a descriptive approach and a prescriptive one.

Harland et al.~\cite{Harland02} show how linear logic enables reasoning about negotiations, encoding agents' goals and what they offer.
Linear logic proofs recognise the negotiation outcomes that satisfy all parties.

Küngas et al.~\cite{Kungas2003,Kungas2004} propose a model of cooperative problem
solving, and use linear logic for encoding agents’ resources, goals and capabilities. Then, each agent determines whether she can solve the problem in isolation. If she cannot, then she starts negotiating with other agents in order to find a cooperative solution.
Partial deduction~\cite{partialdeduction} is used to derive possible deals. 
The authors of~\cite{Kungas2006,Kungas2008} extend their work by considering coalition formation.

Porello et al.~\cite{Porello2010} target distributed resource allocation.
They encode resource ownership and transfers, as well as valuation functions representing user preferences in (various fragments of) linear and affine logic.
They show how logic proofs discriminate mutually satisfactory exchanges that increase the value of the assignment for every user, thus recovering a notion of social welfare in terms of Pareto optimality.
Since the valuation functions of users used to decide exchanges are assumed  known, offers and negotiation are not modeled.
They prove that any sequence of
individually rational deals will always converge to an allocation with
maximal social welfare, as known from~\cite{Sandholm2002}.
In contrast, we directly encode the user exchange policies. 
Afterward we further constrain the exchanges compliant with policies with valuation functions obtaining deals.
Moreover, we extend the computational fragment of linear logic with a contractual implication and we prove decidability results.

Troquard~\cite{Troquard2018} models the interaction of resource-conscious agents who share resources to achieve their goals in cooperative games.
Algorithms are proposed for deciding whether a group of agents can form a coalition and act together in a way that
satisfies them all. 
Various problems concerning cooperative games are modelled in suitable fragments of linear and affine logic and their computational complexity is discussed.
Our focus is instead on resource exchanges, and our context is a mixture of cooperative and competitive behaviour.
In a subsequent work, Troquard~\cite{Troquard2020} studies how a central authority can modify the set of Nash equilibria in a cooperative game by redistributing the initial assignment of resources to agents. 
The complexity of this optimization problem is discussed in terms of the chosen (fragment of) resource-sensitive logic.

\paragraph*{Contractual logics}
We formalised the contractual aspects following the pioneering PCL proposed by Bartoletti and Zunino~\cite{BZ}, which is a logic for modelling contractual reasoning. 
Our operator $\linearcontract$ is actually a linear version of their $\contract$.
The main difference with respect to PCL is that from the premise $p \contract p', p' \contract p$ one can derive $p$, $p'$, and $p \land p'$. 
Instead, in \MuACL\ only the conjunction $p \otimes p'$ can be derived from the premise $p \linearcontract p', p' \linearcontract p$.
The syntactic form of \MuACL\ sequents is inspired by Kanovich~\cite{Kanovich94}, who proposed  a computational fragment of linear logic  for reasoning on computations with consumable resources.

\section{Conclusions and future work}\label{sec:col-conclude}
We introduced exchange environments as a formal model for scenarios where agents join coalitions and exchange resources to achieve individual and collective goals.
Our model is a transition system where states are resource allocations to agents and transitions are labelled by the exchanged resources.
Moreover, we proposed exchange policies to regulate competitiveness and cooperation: agents prescribe in isolation what they offer and what they require in return.

We characterised the notion of agreement as a resource exchange where the policies of all the involved agents are met. 
Since agreements are often circular, checking an exchange to be such is crucial and hard.
For that, we extended the computational fragment of linear logic with a new operator that handles both contracts and circularity.
The resulting logic, called \MuACL, is decidable (\autoref{thm:MuACLdecide1}), so checking that an exchange is an agreement is reduced to finding a proof for its encoding in \MuACL.
We also modeled the case in which an agent incurs a temporary debt that is paid with resources she can obtain during the same exchange.
Extending our logic with the cut rule sufficed to handle these cases, still maintaining decidability (\autoref{thm:MuACLdecide}).

In addition, we formalised when an agreement is beneficial to all the agents of a coalition, dubbed a deal, allowing agents to assign a utility value to resource allocations.
Checking that a resource exchange is a deal consists in showing that the utility value of the offered resources is less than that of those required back (\autoref{th:pol-deal}).
We also showed that there exists a rational policy, i.e. accepting all and only the deals for a coalition (\autoref{th:exists-rational}), and characterized those that guarantee to reach a Pareto increment (\autoref{th:Pareto}).

Future work includes extending \MuACL\ with universal quantifiers, disjunction and negation to express richer policies in a handy and concise manner.
Another line of development concerns creating, disposing and transforming resources in exchanges, actions that are already available in linear logic.
Finally, we plan to investigate to apply our model to describe real scenarios, such as exchanges of crypto-assets in blockchain systems.

\begin{ack}
Work partially supported by projects SERICS (PE00000014) and 
PRIN AM$\forall$DEUS (P2022EPPHM) under the MUR-PNRR funded by the European Union --- NextGenerationEU.
\end{ack}

\bibliography{references}

\appendix

\clearpage

\section{Value Functions}\label{app:vf}

\polanddeal*
\begin{proof}
First, notice that, by definition, $v_\usr(\st') = v_\usr(\st) + W(\usr, \exc)$
for all $\usr$s when $\st \xrightarrow{\exc} \st'$.

Then, we proceed by induction on the definition of $\vDash$.
The base case, $\pol_{\coal} \vDash \emptyset \triangleleft \emptyset$ is trivial.
Assume 
$\pol_{\coal} \vDash \exc \triangleleft \exc'$, 
$\pol_{\coal} \vDash \exc'' \triangleleft \exc'''$, and therefore 
$\pol_{\coal} \vDash (\exc \uplus \exc'') \triangleleft (\exc' \uplus \exc''')$.
By induction hypothesis, $\exc \uplus \exc'$ and $\exc''\uplus \exc'''$ are deals.
Hence, $W(\usr, \exc \uplus \exc'')$ and $W(\usr, \exc' \uplus \exc''')$ are all grated than zero for any $\usr \in \coal$.

We prove now that $\exc \uplus \exc' \uplus \exc'' \uplus \exc'''$ is a deal.
Assume $\st \xrightarrow{\exc \uplus \exc' \uplus \exc'' \uplus \exc''' } \st'$.
Then 
\[
v_\usr(\st') = v_\usr(\st) + W(\usr, \exc \uplus \exc' \uplus \exc'' \uplus \exc''')
\]
We conclude by noticing that $W(\usr, \exc \uplus \exc' \uplus \exc'' \uplus \exc''')$ is the same as $W(\usr, \exc \uplus \exc'') + W(\usr, \exc' \uplus \exc''')$
\end{proof}

\existsrational*
\begin{proof}
The proof amount at showing that for each $\coal$, a finite set of exchanges $\mathbb{E}$ exists such that all the 
deals for $\coal$ can be obtained by disjoint union of (possibly replicated) elements of $\mathbb{E}$.
The result follows from the fact that the set of transitions in an exchange environment is finite, and therefore also all the exchanges that may appear as labels are finite (i.e. the ones where each resource $\res$ does not appear more than $q(\res)$ times).
\end{proof}

\pareto*
\begin{proof}
Note that, since valuation functions are isolated, $W(\usr,\exc) = W(\usr,\exc\downarrow_\usr)$ for any $\usr$.

Assume $\exc$ is an agreement, then it is $\exc = \biguplus_{\coal} \exc_\coal = \biguplus_{\coal} \exc_\coal'$.
Since policies are rational, $W(\usr, (\exc_\coal \uplus \exc_{\coal}')\downarrow_\usr) \geq 0$ for any $\usr \in \coal$.
Finally, it holds by construction that for any $\usr$, $\exc\downarrow_\usr = \biguplus_{\coal} (\exc_\coal \uplus \exc_\coal')\downarrow_\usr$.
Then, for any $\usr$
\begin{align*}
W(\usr, \exc) =
W(\usr, \exc\downarrow_\usr) =
\sum_{C} W(\usr, (\exc_\coal \uplus \exc_{\coal}')\downarrow_\usr) \geq 0
\end{align*}
\end{proof}

\section{Normalization}\label{app:norm}
\begin{figure*}[t]
\begin{gather*}
\prftree[r]
	{($\Omega$-$\otimes$-left)}
	{\omega, \omega', \Phi \vdash \phi}
	{\omega \otimes \omega', \Phi \vdash \phi}
\ \ \
\prftree[r]
	{($\Xi$-$\otimes$-left)}
	{\xi, \xi', \Phi \vdash \phi}
	{\xi \otimes \xi', \Phi \vdash \phi}
\ \ \
\prftree[r]
	{($\Theta$-$\otimes$-left)}
	{\theta, \theta', \Phi \vdash \phi}
	{\theta \otimes \theta', \Phi \vdash \phi}
\ \ \
\prftree[r]
	{($\Delta$-$\otimes$-left)}
	{\delta, \delta', \Phi \vdash \phi}
	{\delta \otimes \delta', \Phi \vdash \phi}
\ \ \
\prftree[r]
	{($\Sigma$-$\otimes$-left)}
	{\sigma, \sigma', \Phi \vdash \phi}
	{\sigma \otimes \sigma', \Phi \vdash \phi}
\end{gather*}
\caption{Specific forms for ($\otimes$-left)}
\vspace{\baselineskip}
\label{fig:otimescases}
\end{figure*}

Note that the syntactic constraints over \MuACL\ sequents causes every application of ($\otimes$-left) to be of one of the forms in~\autoref{fig:otimescases}.

In the following we write \MuACLs\ for \MuACL\ augmented with the (Cut) rule.
We define the following sets of rules
\begin{align*}
R_1 &= \{ \text{(Weak), (Cont), ($!$-left), ($\Omega$-$\otimes$-left), ($\Xi$-$\otimes$-left), ($\Theta$-$\otimes$-left)} \}\\
R_2 &= \{ \text{($\linearcontract$-split)} \}\\
R_3 &= \{ \text{($\otimes$-right), ($\Delta$-$\otimes$-left), (Ax), ($I$-right), ($I$-left)} \}\\
R_4 &= R_3 \cup \{ \text{($\multimap$-left), ($\Sigma$-$\otimes$-left), ($\linearcontract$-left)} \}\\
R_5 &= R_4 \cup \{ \text{(Cut)} \}
\end{align*}
Moreover, we write $R_{n_1, \dots, n_m}$ with $n_i \in \{1,2,3,4,5\}$ for $\bigcup_{i = 1}^n R_{n_i}$.
Recall that we use a double line to represent multiple applications of the same rule.

First of all, note that every proof for a sequent $\delta \vdash \delta'$ already uses only the rules in $R_3$, which are the legal ones for $\Pi_3$ and $\Pi_3'$, as all the others cannot be applied due to syntactic constraints.

We then start with some auxiliary lemmata about reordering rules in \MuACLs, where we ignore the first two premises of ($\linearcontract$-left).
\begin{lemma}\label{thm:move1}
	Any \MuACLs\ derivation can be rewritten as an equivalent derivation where no rule
	in $R_1$ is applied before one in $R_{2,5}$, and the new derivation uses (Cut) only if the original one does.
\end{lemma}
\begin{proof}
	We proceed by cases, showing that any application of $r \in R_1$ followed by some $r' \in R_{2,5}$ can be rewritten as an equivalent derivation that satisfies our condition.
	All the cases are straightforward, consider for example  $r =$ (Weak) and $r' =$ (Cut) and assume that $r$ is applied in the derivation of the left premise:
	\[
	\prftree[r]{(Cut)}
	{
		\prftree[r]
		{(Weak)}
		{\Phi \vdash \sigma}
		{\Phi, !\phi \vdash \sigma}
	}
	{
		\prftree[noline]
		{\Phi', \sigma \vdash \sigma'}
	}
	{\Phi, \Phi', !\phi \vdash \sigma}
	\]
	Then, swap the rules as follows:
	\[
	\prftree[r]{(Weak)}
	{
		\prftree[r]{(Cut)}
		{\Phi \vdash \sigma}
		{\Phi', \sigma \vdash \sigma'}
		{\Phi, \Phi' \vdash \sigma}
	}
	{\Phi, \Phi', !\phi \vdash \sigma}
	\]
	Similarly if $r$ applied to the derivation of the right premise.
\end{proof}

\begin{lemma}\label{thm:move2}
	Any \MuACLs\ derivation can be rewritten as an equivalent derivation where no rule
	in $R_2$ is applied before one in $R_{5}$, and the new derivation uses (Cut) only if the original one does.
\end{lemma}
\begin{proof}
	We proceed by cases, showing that any application of $(\linearcontract\text{-split}) \in R_2$ followed by some $r \in R_{5}$ can be rewritten as an equivalent derivation that satisfies our condition.
	Let $r = $($\linearcontract\text{-split}$), then we can rewrite
		\[
		\prftree[r]{($\linearcontract$-left)}
		{\delta_0 \vdash \delta_0'}
		{\delta_0' \vdash \delta_0}
		{
			\prftree[r]{($\linearcontract$-split)}
			{\Phi, \delta_0', (\delta \otimes \delta'') \linearcontract (\delta' \otimes \delta''') \vdash \sigma}
			{\Phi, \delta_0', \delta \linearcontract \delta', \delta'' \linearcontract \delta''' \vdash \sigma}
		}
		{\Phi, \delta_0 \linearcontract \delta_0', \delta \linearcontract \delta', \delta'' \linearcontract \delta''' \vdash \sigma}
		\]
		as
		\[
		\prftree[r]{($\linearcontract$-split)}
		{
			\prftree[r]{($\linearcontract$-left)}
			{\delta_0 \vdash \delta_0'}
			{\delta_0' \vdash \delta_0}
			{\Phi, \delta_0', (\delta \otimes \delta'') \linearcontract (\delta' \otimes \delta''') \vdash \sigma}
			{\Phi, \delta_0 \linearcontract \delta_0', (\delta \otimes \delta'') \linearcontract (\delta' \otimes \delta''') \vdash \sigma}
		}
		{\Phi, \delta_0 \linearcontract \delta_0', \delta \linearcontract \delta', \delta'' \linearcontract \delta''' \vdash \sigma}
		\]
	All the other cases are straightforward, since no other rule in $R_{5}$ effectively modifies $\Theta$.
\end{proof}

\begin{lemma}\label{thm:move3}
	Any \MuACLs\ derivation that only uses rules in $R_5$ can be rewritten as an equivalent derivation where no rule is applied after ($\linearcontract$-left).
	In addition, the equivalent derivation uses (Cut) only if the original one does.
\end{lemma}
\begin{proof}
	We proceed by cases, showing that any application of ($\linearcontract$-left) followed by some other $r \in R_{5}$ can be rewritten as an equivalent derivation that satisfies our condition.
	All the cases are straightforward, since no other rule acting on $\theta$ is in $R_5$.
\end{proof}

\begin{lemma}\label{thm:move4}
	Any derivation composed by two subsequent applications of ($\linearcontract$-left) can be rewritten as an application of ($\linearcontract$-left) followed by ($\linearcontract$-split).
\end{lemma}
\begin{proof}
	The derivation
	\[
	\prftree[r]{($\linearcontract$-left)}
	{\delta \vdash \delta'}
	{\delta' \vdash \delta}
	{
		\prftree[r]{($\linearcontract$-left)}
		{\delta'' \vdash \delta'''}
		{\delta''' \vdash \delta''}
		{
			\prftree[noline]
			{\Phi, \delta', \delta''' \vdash \sigma}
		}
		{\Phi, \delta', \delta'' \linearcontract \delta''' \vdash \sigma}
	}
	{\Phi, \delta \linearcontract \delta', \delta'' \linearcontract \delta''' \vdash \sigma}
	\]
	can be rewritten as
	\[
	\prftree[r]{($\linearcontract$-split)}
	{
		\prftree[r]{($\linearcontract$-left)}
		{
			\prftree[noline]
			{\Pi}
			{\delta \otimes \delta'' \vdash \delta' \otimes \delta'''}
		}
		{
			\prftree[noline]
			{\Pi'}
			{\delta' \otimes \delta''' \vdash \delta \otimes \delta''}
		}
		{
			\prftree[noline]
			{\Phi, \delta', \delta''' \vdash \sigma}
		}
		{\Phi, \delta \otimes \delta'' \linearcontract \delta' \otimes \delta''' \vdash \sigma}
	}
	{\Phi, \delta \linearcontract \delta', \delta'' \linearcontract \delta''' \vdash \sigma}
	\]
	where $\Pi$ is 
	\[
	\prftree[r]{($\Delta$-$\otimes$-left)}
	{
		\prftree[r]{($\otimes$-right)}
		{\delta \vdash \delta' }
		{\delta'' \vdash \delta'''}
		{\delta, \delta'' \vdash \delta' \otimes \delta'''}
	}
	{\delta \otimes \delta'' \vdash \delta' \otimes \delta'''}
	\]
	and, similarly, $\Pi'$ is
	\[
	\prftree[r]{($\Delta$-$\otimes$-left)}
	{
		\prftree[r]{($\otimes$-right)}
		{\delta' \vdash \delta }
		{\delta''' \vdash \delta''}
		{\delta', \delta''' \vdash \delta \otimes \delta'' }
	}
	{\delta' \otimes \delta'''\vdash \delta \otimes \delta'' }
	\]
\end{proof}

We also define normal forms for \MuACLs.
\begin{definition}[\MuACLs\ Normal Proofs]\label{def:normalforms}
	A \MuACLs\ proof for an initial sequent is \emph{normal} if 
	it can be decomposed in either form, where 
	$\Pi_i$ and $\Pi_i'$ only use rules in $R_i$.
	\begin{center}
		\begin{tabular}{c}
			\begin{tabular}{c}
				\prfsummary[$\Pi_1$]
				{
						\prftree
						{\Pi_5}
						{\Delta, \Sigma \vdash \sigma}
				}
				{\Omega, \Xi, \Sigma \vdash \sigma}\\[0.2cm]
				\textit{normal form 1}
			\end{tabular}
			
			\begin{tabular}{c}
				\hspace{0.5cm}	
				\prfsummary[$\Pi_1$]
				{
						\prfsummary[$\Pi_2$]
						{
							\prftree[r]
							{($\linearcontract$-left)}
							{\Pi_3}
							{\Pi_3'}							
							{
								\prftree
								{\Pi_5}
								{\Delta, \delta, \Sigma \vdash \sigma}
							}
							{\theta, \Delta, \Sigma \vdash \sigma}
						}
						{
							{\Theta, \Delta, \Sigma \vdash \sigma}
						}
				}
				{\Omega, \Xi, \Sigma \vdash \sigma}\\[0.2cm]
				\textit{normal form 2}
			\end{tabular}
		\end{tabular}
	\end{center}
\end{definition}

The following result prove that \MuACLs\ proofs in normal form correctly characterize validity of initial sequents.
\begin{restatable}[\MuACLs\ Normal Proofs]{theorem}{normalforms}\label{thm:col-normal-ncr-s}
	For any $\Omega, \Sigma, \sigma$, the initial sequent $\Omega; \Sigma \vdash \sigma$ is valid in 
	\MuACLs\ if and only if a normal proof\  $\,\Pi$ exists for $\Omega; \Sigma \vdash \sigma$.
\end{restatable}
\begin{proof}
First, we apply~\autoref{thm:move1} and \ref{thm:move2}, obtaining $\Pi_1$, $\Pi_2$ and $\Pi'$, the remaining part of the proof.
Then, \autoref{thm:move3} allows to move all the uses of ($\linearcontract$-left) in the bottom of $\Pi'$, and, by~\autoref{thm:move4} we finally simplify them into a single application of the rule.
\end{proof}

We recover the following a corollary for the \MuACL\ fragment of \MuACLs.
\normalform*
\begin{proof}
Follows from~\autoref{thm:col-normal-ncr-s} and by noticing that the used lemmas guarantee that no (Cut) rule is introduced.
\end{proof}

\section{Decidability}\label{app:proofs}

We now prove our main result, i.e., that \MuACL\ and \MuACLs\ are decidable.
We first focus on \MuACLs, as the case for \MuACL\ can be derived easily.

By~\autoref{thm:col-normal-ncr}, we only consider normal proofs.
In the following, we verify if a proof in the normal form 1 exists for an initial sequent, then we show how to reduce the normal form 2 case 
to the normal form 1 case.

\subsection{Solving the Normal Form 1}

\begin{lemma}[Normal form 1 decidability]\label{thm:fstnfdecide}
	An always-terminating algorithm exists that decides if an initial sequent is provable in \MuACLs\ using a proof in the normal form 1.
\end{lemma}
\begin{proof}
	This result derives from a similar one by Kanovich~\cite{Kanovich94}, which is stated for a computational fragment of linear logic that coincides with our  sequents when $\Xi = \Theta = \emptyset$.
	Kanovich considers \emph{simple products}, i.e., linear conjunctions of atomic predicates (our $\sigma$); \emph{Horn-implications}, i.e., linear implications of simple products (our $\delta$); and \emph{!-Horn-implications}, i.e., Horn implications preceded by $!$ (our $\Xi$).
	Moreover, he defines \emph{!-Horn-sequents}, i.e., sequents with !-Horn-implications, Horn-implications and simple products as left parts and simple products as right part (our initial sequents).

	Finally, the problem of checking the validity of a !-Horn sequent is reduced in~\cite{Kanovich94} to reachability in Petri Nets, which can be decided using the algorithm proposed in~\cite{PetriReach}.
	Roughly, atomic proposition corresponds to places of the Petri Net, and linear implication to transitions.
	The number of tokens in a given place represents the occurrences of the corresponding atomic proposition, and changes according to linear implications that we can use ad libitum.
\end{proof}

We recover~\autoref{thm:fstnfdecide-ncr} as a special case of the Lemma above.
\fstnfdecidencr*
\begin{proof}
It suffices to adapt Kanovich's encoding by forbidding the outcome of a linear implication to be used in subsequent transitions.
For each atomic proposition $p$ we define two places of the Petri Net $p_s$ and $p_t$.
For each linear implication $\delta = \sigma \multimap \sigma'$ such that $!\delta$ is a subterm in $\Omega$, we define a transition in the Petri Net
which consumes the tokens from $p_s$, with $p \in \sigma$ and produces the ones for $p'_t$ for $p' \in \sigma'$.
Moreover, we add transitions from each $p_s$ to $p_t$ allowing atomic propositions to be taken as they are (through the (Ax) rule). 
Note that we can still use linear implications ad libitum, but we cannot reuse their outcome as an input for others linear implications.
\end{proof}

\subsection{Reducing the Normal Form from 2 to 1}

Consider a semiring module $M$ over the set of natural numbers $\mathbb{N}$ with subformulas of any \MuACL\ predicate $\phi$ as its basis (we can safely reduce to the finite set of the ones appearing in the given sequent we are considering).
We call the elements $\bar x$ of $M$ vectors for simplicity, and write $\bar x(\phi)$ for the natural number associated in $\bar x$ with the bases element $\phi$.
Given a set of vectors $A$, we let $\mathit{span}(A)$ be the set of all linear combinations of elements in $A$.

Given a $\delta$, we define the vector $\bar u_{\delta}$ associating each linear implication $\res@\usr \multimap \res@\usr$ with the number of its occurrences in $\delta$.
This correspondence is an isomorphism up-to commutativity and associativity of $\otimes$ and unity of $I$, and is extended to multisets $\Delta$.

Note that the derivation $\Pi_1$ of the normal form 1 essentially decides how many occurrences of each $\delta$ with $!\delta \in \Xi$ are taken for constructing $\Delta$.
Let $A_{\Xi}$ be the set of vectors $u_{\delta}$ with $!\delta \in \Xi$, then the following holds.
\begin{lemma}\label{thm:pi1fornf1}
A derivation $\Pi_1$ exists from $\Delta, \Sigma \vdash \sigma$ to $\Xi, \Sigma \vdash \sigma$ if and only if $\bar u_{\Delta} \in \mathit{span}(A_{\Xi})$.
\end{lemma}
\begin{proof}
	By induction on the applicable rules.
\end{proof}

Consider now the normal form 2.
Given $\theta$ and $\Theta$, we let $\bar u_{\theta}$ and $\bar u_{\Theta}$ be the vectors associating each $\delta$ with the number of its occurrences in $\theta$ and $\Theta$ respectively.
Then, for a given $\Omega$, we define $B_{\Omega}$ 
as the set of vectors $u_{\theta}$ with $!\theta \in \Omega$.
\begin{lemma}\label{thm:pi1fornf2}
A derivation $\Pi_1$ exists from $\Theta, \Delta, \Sigma \vdash \sigma$ to $\Omega, \Xi, \Sigma \vdash \sigma$ if and only if 
$\bar u_{\Delta} \in \mathit{span}(A_{\Xi})$ and $\bar u_{\Theta} \in \mathit{span}(B_{\Omega})$.
\end{lemma}
\begin{proof}
	By induction on the applicable rules.
\end{proof}

It is straightforward to check that $\Pi_2$ exists if and only if $\theta$ is the result of collecting all the left and right parts of $\linearcontract$ in the formulas appearing in $\Theta$.
\begin{lemma}\label{thm:colapp-merge}
	A derivation that only uses ($\linearcontract$-split) exists from $\theta, \Delta, \Sigma \vdash \sigma$ to $\Theta, \Delta, \Sigma \vdash \sigma$
	if and only if $\theta = \delta \linearcontract \delta'$ with
	\[
	\delta = \bigotimes_{\delta_i \linearcontract \delta_i' \in \Theta} \delta_i \quad \text{ and } \quad \delta' = \bigotimes_{\delta_i \linearcontract \delta_i' \in \Theta} \delta_i'
	\]
\end{lemma}
\begin{proof}
	Derives from the fact that ($\linearcontract$-split) preserves both the multisets of instances of $\delta$ that appears to the left of $\linearcontract$ and the multiset of the instances of $\delta$ that appears to the right of $\linearcontract$.
\end{proof}

Moreover, $\Pi_3$ and $\Pi_3'$ exist if and only if $u_{\delta} = u_{\delta'}$.
\begin{lemma}\label{thm:deltaproof}
For any $\delta, \delta'$, $\delta \vdash \delta'$ and $\delta' \vdash \delta$ both hold if and only if $u_{\delta} = u_{\delta'}$.
\end{lemma}
\begin{proof}
	By induction on the applicable rules.
\end{proof}
We let $L(\bar x)$ and $R(\bar x)$ be linear functions returning the linear implications appearing to the left and right of $\linearcontract$ respectively of $\bar x$.
\begin{align*}
L(\bar{x}) &= \bar y \ \ \text{ s.t. } \\
	 &\ \bar y (\res@\usr \multimap \res'@\usr') = 
\sum_{\theta = \delta \linearcontract \delta'} \bar{x}(\theta) \cdot \bar u_{\delta}(\res@\usr \multimap \res'@\usr')\\
R(\bar{x}) &= \bar y \ \ \text{ s.t. } \\
	&\ \bar y (\res@\usr \multimap \res'@\usr') = 
\sum_{\theta = \delta \linearcontract \delta'} \bar{x}(\theta) \cdot \bar u_{\delta'}(\res@\usr \multimap \res'@\usr')
\end{align*}

\begin{lemma}
The derivations $\Pi_3$ and $\Pi_2$ from $\Delta, \delta, \Sigma \vdash \sigma$ to $\Theta, \Delta, \Sigma \vdash \sigma$ if and only if 
$L(\bar u_\Theta) - R(\bar u_\Theta) = 0$ and $R (\bar u_\Theta) = u_\delta$.
\end{lemma}
\begin{proof}
By \autoref{thm:colapp-merge} and \ref{thm:deltaproof}, and by observing that the first two premises of ($\linearcontract$-left) coincide with the subformulas at the left and right of $\linearcontract$ in $\theta$, with $\delta$ in the third premise being the right part.
\end{proof}

By applying the previous results, we can conclude the following. 
\sftoff*
\begin{proof}
By combining the previous results, the derivations $\Pi_1$, $\Pi_2$ and $\Pi_3$ exist if and only if
$\bar u_{\Delta} \in \mathit{span}(A_{\Xi})$, $\bar u_{\delta} \in \mathit{span}(R(B_{\Omega}))$ with coefficients $\bar{x}$ such that $L(B_{\Omega}) - R(B_{\Omega}) (\bar{x}) = 0$.
Thanks to the Hilbert basis theorem~\cite{HB}, we can represent the set of solutions $\bar{x}$ of the equation above as $\mathit{span}(H_{\Omega})$ for some $H_\Omega$ that can be computed using~\cite{computeHB}.
Then, we can reformulate our conditions in terms of just linear implications, requiring $\bar u_{\delta} \in \mathit{span}(R(B_{\Omega}) \cdot H_\Omega)$.
Since $R(B_{\Omega}) \cdot H_\Omega$ is a linear combinations of vectors that associates with weights different from $0$ only propositions that are linear implications, it must hold that $R(B_{\Omega}) \cdot H_\Omega = A_{\Xi'}$ for some $\Xi'$.

The result then follows by~\autoref{thm:pi1fornf1}
\end{proof}

\begin{theorem}[\MuACLs\ decidability]\label{thm:MuACLsdecide}
	An always-terminating algorithm exists that decides if an initial sequent is valid in \MuACLs.
\end{theorem}
\begin{proof}
	\autoref{thm:col-sftoff} reduces the problem of finding a proof in the normal form 2 to finding a proof in the normal form 1, which is proved decidable in~\autoref{thm:fstnfdecide}.
\end{proof}

\MuACLsdec*
\begin{proof}
	The result follows directly by \autoref{thm:col-sftoff}, which reduces the problem of finding a proof in the normal form 2 to finding a proof in the normal form 1, which is proved decidable in~\autoref{thm:fstnfdecide-ncr}.
	Notice that the (Cut) rule is not introduced in the reduction, as $\Pi_4$ remains the same.
\end{proof}

\section{Encoding}\label{app:compile}

We address separately correctness and completeness.

\subsection{Correctness}

Hereafter, we only consider proofs of initial sequents that are the result of our encoding. 
We start by noticing that the normal forms for such proofs have specific constraints, as shown in the following lemma.
\begin{lemma}[MuAC normal form]\label{thm:specnormal}
	A proof $\Pi$ for a sequent $\biguplus_{\usr \in \Usr} \semden{\pol_\coal}, \semden{\MuACstate} \vdash \semden{\MuACstate'}$ in \MuACLs\ or \MuACL\ is \emph{normal} if and only if it can be decomposed in either form
	\begin{center}
		\begin{tabular}{c}
			\begin{tabular}{c}
				\prfsummary[$\Pi_1$]
				{
					\prftree
					{\Pi_4}
					{\semden{\MuACstate} \vdash \semden{\MuACstate'}}
				}
				{\biguplus_{\usr \in \Usr} \semden{\pol_\coal}, \semden{\MuACstate} \vdash \semden{\MuACstate'}}\\[0.2cm]
				\textit{normal form 1}
			\end{tabular}
\\

\\

			\begin{tabular}{c}
				\prfsummary[$\Pi_1$]
				{
						\prfsummary[$\Pi_2$]
						{
							\prftree[r]{($\linearcontract$-left)}
							{\Pi_3}
							{\Pi_3'}
							{
								\prftree
								{\Pi_4}
								{\delta, \semden{\MuACstate} \vdash \semden{\MuACstate'}}
							}
							{\theta, \semden{\MuACstate} \vdash \semden{\MuACstate'}}
						}
						{
							{\Theta, \semden{\MuACstate} \vdash \semden{\MuACstate'}}
						}
				}
				{\biguplus_{\usr \in \Usr} \semden{\pol_\coal}, \semden{\MuACstate} \vdash \semden{\MuACstate'}}\\[0.2cm]
				\textit{normal form 2}
			\end{tabular}
		\end{tabular}
	\end{center}
\end{lemma}
\begin{proof}
	Take the normal form 1 of~\autoref{def:normalforms}.
	We must show that $\Delta = \emptyset$.
	Clearly, this is the case, since no $!\delta$ is in $\biguplus_{\usr \in \Usr} \semden{\pol_\coal}$.
	
	For the same reason, in a proof in the normal form 2 of~\autoref{def:normalforms}, $\Delta'$ must be equal $\emptyset$.
\end{proof}

Proofs in the normal form 1 are trivial because they correspond to proofs where the state does not change (and thus both the correctness and completeness in this case follow trivially).
Hence in the following we will only consider proofs in the normal form 2, i.e., the ones corresponding to nonempty exchanges.

\begin{lemma}\label{thm:acceptderivetofair}
	Consider a proof as in~\autoref{thm:specnormal}.
	If $\Pi_1, \Pi_2$, $\Pi_3$ and $\Pi_3'$ exist, then $\delta = \semden{\exc}$ for some agreement $\exc$.
\end{lemma}
\begin{proof}
	First, note that $\semden{\exc'} \linearcontract \semden{\exc} \in \Theta$ implies $\exc \triangleleft \exc' \in \pol_\coal$ for some $\coal$.
	Let $\Theta_\coal$ be the multiset of all the $\semden{\exc'} \linearcontract \semden{\exc}$ obtained from $\coal$, then $\pol_\coal \vdash \exc_\coal \triangleleft \exc_\coal'$, with $\exc_\coal$ ($\exc_\coal'$ resp.) the tensor product of all the left (right resp.) parts of $\Theta_\coal$.
	Hence, the decomposition of $\Theta$ is such that every $\coal$ accepts $\exc \triangleleft \exc'$ such that $\semden{\exc \triangleleft \exc'} = \Theta_\coal$.
	Finally, the left and right parts of $\theta$ coincides thanks to the left premises of ($\linearcontract$-left).
\end{proof}

\begin{lemma}\label{thm:deltaexctofair}
	For any $\MuACstate, \MuACstate'$ and $\exc$, if $\semden{\exc}, \semantics{\MuACstate} \vdash \semantics{\MuACstate'}$ is valid in \MuACL\ then $\MuACstate \xrightarrow{\exc} \MuACstate'$.
\end{lemma}
\begin{proof}
By induction on the rules of \MuACL.
We can ignore rules that are not applicable due to the form of the sequent.
The property trivially holds for (Ax) with $\exc = \emptyset$ and for ($\otimes$-left).
Consider the rule ($\otimes$-right), the property follows by the induction hypothesis since 
$\st_1 \xrightarrow{\exc_1} \st_1'$ and $\st_2 \xrightarrow{\exc_2} \st_2'$ implies $(\st_1 \uplus \st_2) \xrightarrow{\exc_1 \uplus \exc_2} (\st_1' \uplus \st_2')$.
Consider the rule ($\multimap$-left), and note that $\semantics{\st} \vdash \semantics{\st'}$ implies $\st = \st'$.
Then, the result follows by definition of $\Delta_{\exc}$.
\end{proof}

\begin{lemma}\label{thm:validitytofairness}
	For any $\st, \st'$, if $\biguplus_{\usr \in \Usr} \semden{\pol_\coal}, \semden{\st} \vdash \semden{\st'}$ is valid in \MuACL, then 
	$\st \rightarrow_{ok} \st'$.
\end{lemma}
\begin{proof}
Follows from~\autoref{thm:acceptderivetofair} and~\ref{thm:deltaexctofair}
\end{proof}

\begin{lemma}\label{thm:deltaexctofaircomp}
	For any $\MuACstate, \MuACstate'$ and $\exc$, if $\semden{\exc}, \semantics{\MuACstate} \vdash \semantics{\MuACstate'}$ is valid in \MuACLs\ then $\MuACstate \xdashrightarrow{\exc} \MuACstate'$.
\end{lemma}
\begin{proof}
By induction on the rules of \MuACLs.
We can ignore rules that are not applicable due to the form of the sequent.
For ($\Sigma$-Ax), ($\otimes$-left), ($\otimes$-right) and ($\multimap$-left) the result follows from~\autoref{thm:deltaexc}.

Consider the rule 
\[
\prftree[r]
{(Cut)}
{\semantics{\exc_1}, \semantics{\st} \vdash \semantics{\st'}\quad}
{\semantics{\exc_2}, \semantics{\st'} \vdash \semantics{\st''}}
{\semantics{\exc_1}, \semantics{\exc_2}, \semantics{\st} \vdash \semantics{\st''}}
\]
By the induction hypothesis, we know that $\st \xdashrightarrow{\exc_{1}} \st'$ and $\st' \xdashrightarrow{\exc_{2}} \st''$.
The result then derives from noticing that $\st \xdashrightarrow{\exc_1 \uplus \exc_{2}} \st''$
$\semantics{\exc_1} \otimes \semantics{\exc_2} = \semantics{\exc_1 \uplus \exc_2}$ and by applying ($\otimes$-left).
\end{proof}

\begin{lemma}\label{thm:validitytofairnesscomp}
	For any $\st, \st'$, if $\biguplus_{\usr \in \Usr} \semden{\pol_\coal}, \semden{\st} \vdash \semden{\st'}$ is valid in \MuACLs, then 
	$\st \dashrightarrow_{ok} \st'$.
\end{lemma}
\begin{proof}
Follows from~\autoref{thm:acceptderivetofair} and~\ref{thm:deltaexctofaircomp}
\end{proof}

\subsection{Completeness and Size of \MuACL\ Proofs}\label{sec:fairnesstovalidity}
We start with an auxiliary result stating a general property of \MuACL\ and \MuACLs, stating that they are monotone with respect to $\Omega$ and $\Xi$.
\begin{proposition}\label{thm:monotonicity}
	For each $\Omega, \Xi, \Phi, \phi, \Phi', \phi'$, 
	if $\Phi \vdash \phi$ is derivable from $\Phi' \vdash \phi'$ in \MuACL\ (or \MuACLs), then also
	$\Omega, \Xi, \Phi \vdash \phi$ is derivable from $\Omega, \Xi, \Phi' \vdash \phi'$ in \MuACL\ (or \MuACLs).
\end{proposition}
\begin{proof}
	By induction on the rules.
\end{proof}

\begin{lemma}\label{thm:acceptderive}
	For any $\exc, \exc'$, if $\pol_\coal \vDash \exc \triangleleft \exc'$ then, for any $\Sigma$ and $\sigma$, 
	a \MuACL\ derivation exists from $\semantics{\pol_\coal}, \semden{\exc'} \linearcontract \semden{\exc}, \Sigma \vdash \sigma$ to 
	$\semantics{\pol_\coal}, \Sigma \vdash \sigma$.
\end{lemma}
\begin{proof}
	By induction on the definition of $\vDash$.
	The base case is trivial.
	Assume that $\pol_{\coal} \vDash \exc \triangleleft \exc'$, $\pol_{\coal} \vDash \exc'' \triangleleft \exc'''$, and therefore
	 $\pol_{\coal} \vDash (\exc \uplus \exc'') \triangleleft (\exc' \uplus \exc''')$.
	
	We can write the following, where $\Pi$, $\Pi'$ exist by the induction hypothesis and~\autoref{thm:monotonicity}.
	\[
	\prfsummary[$\Pi$]
	{
		\prfsummary[$\Pi'$]
		{
			\prftree[noline]
			{\semantics{\pol_\coal}, \semden{\exc'} \linearcontract \semden{\exc}, \semden{\exc'''} \linearcontract \semden{\exc''}, \Sigma \vdash \sigma}
		}
		{\semantics{\pol_\coal}, \semden{\exc'} \linearcontract \semden{\exc}, \Sigma \vdash \sigma}
	}
	{\semantics{\pol_\coal}, \Sigma \vdash \sigma}
	\]
	Then we can apply ($\linearcontract$-split) to the top of the derivation, obtaining $\semantics{\pol_\coal}, (\semden{\exc'} \otimes \semden{\exc'''}) \linearcontract (\semden{\exc} \otimes \semden{\exc''}), \Sigma \vdash \sigma$.
	
	Then, it suffices noticing that $\semden{\exc} \otimes \semden{\exc'} = \semden{\exc \uplus \exc'}$.
\end{proof}

\begin{lemma}\label{thm:fairderive}
	For any agreement $\exc$, and for any $\Sigma, \sigma$, a derivation exists from $\semantics{\exc}, \Sigma \vdash \sigma$ to
	$\biguplus_{\usr \in \Usr} \semden{\pol_\coal}, \Sigma \vdash \sigma$.
\end{lemma}
\begin{proof}
	Assume $\exc$ is an agreement, then $\exc_\coal$ and $\exc_\coal'$ exists, and $\pol_\coal \vDash \exc_C \triangleleft \exc_C'$ for any $\coal$.
	Then, by~\autoref{thm:acceptderive} and~\autoref{thm:monotonicity}, the following derivation exists
	\[
	\prfsummary
	{
		\prftree[r]{($\linearcontract$-left)}
		{
			\prftree
			{\Pi}
			{\delta \vdash \delta'}
		}
		{
			\prftree
			{\Pi'}
			{\delta' \vdash \delta}
		}
		{\delta, \Sigma \vdash \sigma}
		{\delta \linearcontract \delta', \Sigma \vdash \sigma}
	}
	{\biguplus_{\usr \in \Usr} \semden{\pol_\coal}, \Sigma \vdash \sigma}
	\]
	with $\delta = \bigotimes_{\coal} \semden{\exc_\coal'}$ and $\delta' = \bigotimes_{\coal} \semden{\exc_\coal}$.
	We conclude by noticing that the existence of $\Pi$ and $\Pi'$ is guaranteed by the fact $\exc$ is an agreement, and
	\[
	\biguplus_{\coal}\exc_\coal = \biguplus_{\coal} \exc_\coal' = \exc \quad \text{and}\quad  \semantics{\exc} = \delta
	\]
\end{proof}

\begin{lemma}\label{thm:deltaexc}
	For any $\MuACstate, \MuACstate'$ and $\exc$, if $\MuACstate \xrightarrow{\exc} \MuACstate'$ then $\semantics{\exc}, \semantics{\MuACstate} \vdash \semantics{\MuACstate'}$ is provable in \MuACL.
\end{lemma}
\begin{proof}
	By induction on the size of $\exc$.
	The case $\exc = \emptyset$ is straightforward.
	Let $\exc$ be $\exc' \uplus \{ \tr \}$ with $\tr = \usr \xmapsto{\res} \usr'$.
	By definition, $\st = \{ (\usr, \res) \} \uplus \st''$, $\st' = \{ (\usr', \res) \} \uplus \st'''$ with $\st'' \xrightarrow{\exc'} \st'''$.
	
	By induction hypothesis, a proof $\Pi$ exists for $\semantics{\exc'}, \semantics{\st''} \vdash \semantics{\st'''}$.
	Then the following is a proof for $\semantics{\exc}, \semantics{\st} \vdash \semantics{\st'}$.
	\[
	\prftree[r]{($\otimes$-left)}
	{
		\prftree[r]{($\otimes$-right)}
		{
			\prftree[r]{($\multimap$-left)}
			{
				\prftree[r]{(Ax)}
				{\res@\usr \vdash \res@\usr}
			}
			{\semantics{\{\tr\}}, \res@\usr \vdash \res@\usr'}
		}
		{
			\prftree{\Pi}
			{\semantics{\exc'},	\semantics{\st''} \vdash \semantics{\st'''}}
		}
		{\semantics{\exc'}, \semantics{\{\tr\}}, \semantics{\st} \vdash \semantics{\st'}}
	}
	{\semantics{\exc}, \semantics{\st} \vdash \semantics{\st'}}
	\]
\end{proof}

\begin{lemma}\label{thm:fairnesstovalidity}
	For all $\st$ and $\st'$, if $\st \rightarrow_{ok} \st'$, then the sequent $\biguplus_{\usr \in \Usr} \semden{\pol_\coal}, \semantics{\st} \vdash \semantics{\st'}$ is valid in \MuACL.
\end{lemma}
\begin{proof}
	By composing the derivations of~\autoref{thm:fairderive} and \autoref{thm:deltaexc}.
\end{proof}

We can now prove the compilation from MuAC to \MuACL\ to be correct and complete.
\correction*
\begin{proof}
	By~\autoref{thm:validitytofairness} and~\autoref{thm:fairnesstovalidity}.
\end{proof}

We investigate now exchange environments with debits.
\begin{lemma}\label{thm:deltaexccmp}
	For each $\st, \st'$ and $\exc$, if $\st \xdashrightarrow{\exc} \st'$
	then $\semantics{\exc}, \semantics{\st} \vdash \semantics{\st'}$ is valid in \MuACLs.
\end{lemma}
\begin{proof}
	By induction on the size of $\exc$.
	The case $\exc = \emptyset$ is straightforward.
	Let $\exc$ be $\exc' \uplus \{ \tr \}$ with $\tr = \usr \xmapsto{\res} \usr'$.
	We consider two cases.
	If $\st = \{ (\usr, \res) \} \uplus \st''$, $\st' = \{ (\usr', \res) \} \uplus \st'''$ with $\st'' \xrightarrow{\exc'} \st'''$, and then the proof is built as in~\autoref{thm:deltaexc}.
	
	Otherwise $\exc' = \exc'' \uplus \{ \tr' \}$ with $\tr' = \usr'' \xmapsto{\res} \usr$, and
	$\st = \{ (\usr'', \res) \} \uplus \st''$, $\st' = \{ (\usr', \res) \} \uplus \st'''$ with $\st'' \xrightarrow{\exc''} \st'''$.
	
	By induction hypothesis, a proof $\Pi$ exists for $\semantics{\exc''}, \semantics{\st''} \vdash \semantics{\st'''}$.

	Then the following is a proof for $\semantics{\exc}, \semantics{\st} \vdash \semantics{\st'}$.
	\[
	\prftree[r,doubleline]{($\otimes$-left)}
	{
		\prftree[r]{($\otimes$-right)}
		{
			\prftree
			{\Pi'}
			{\semantics{\{\tr\}}, \semantics{\{\tr'\}}, \res@\usr'' \vdash \res@\usr'}
		}
		{
			\prftree{\Pi}
			{\semantics{\exc''},	\semantics{\st''} \vdash \semantics{\st'''}}
		}
		{\semantics{\exc''}, \semantics{\{\tr\}}, \semantics{\{\tr'\}}, \semantics{\st} \vdash \semantics{\st'}}
	}
	{\semantics{\exc}, \semantics{\st} \vdash \semantics{\st'}}
	\]
	with $\Pi'$ defined as
	\[
	\prftree[r]{(Cut)}
	{
		\prftree[r]{($\multimap$-left)}
		{
			\prftree[r]{(Ax)}
			{\res@\usr'' \vdash \res@\usr''}
		}
		{\semantics{\{\tr'\}} \res@\usr'' \vdash \res@\usr}
	}
	{
		\prftree[r]{($\multimap$-left)}
		{
			\prftree[r]{(Ax)}
			{\res@\usr \vdash \res@\usr}
		}
		{\semantics{\{\tr\}} \res@\usr \vdash \res@\usr'}
	}
	{\semantics{\{\tr\}}, \semantics{\{\tr'\}}, \res@\usr'' \vdash \res@\usr'}
	\]
\end{proof}

\begin{lemma}\label{thm:fairnestovaliditycomp}
	For any $\st$ and $\st'$, if $\st \dashrightarrow_{ok} \st'$, then the sequent $\biguplus_{\usr \in \Usr} \semden{\pol_\coal}, \semantics{\st} \vdash \semantics{\st'}$ is valid in \MuACLs.
\end{lemma}
\begin{proof}
	By composing the derivations of~\autoref{thm:fairderive} and \autoref{thm:deltaexccmp}.
\end{proof}

The following theorem states the equivalence between provability in \MuACLs\ and agreement-labelled transitions in an exchange environment with debts.
\Logiccorrectnessstar*
\begin{proof}
	By~\autoref{thm:validitytofairnesscomp} and~\autoref{thm:fairnestovaliditycomp}.
\end{proof}

\end{document}